\documentclass[11pt,a4paper,leqno]{article}

\usepackage{amsmath,amssymb,amsthm}
\usepackage{complexity}
\usepackage{tcolorbox}
\usepackage{booktabs}
\usepackage{enumitem}
\usepackage{thm-restate}
\usepackage{authblk}
\usepackage{xspace}
\usepackage{url}
\usepackage{hyperref}
\newtheorem{theorem}{Theorem}
\newtheorem{lemma}[theorem]{Lemma}
\newtheorem{lemma*}[theorem]{Lemma}
\newtheorem{proposition*}[theorem]{Proposition}
\newtheorem{claim}[theorem]{Claim}
\newtheorem{example}[theorem]{Example}
\newtheorem{proposition}[theorem]{Proposition}
\newtheorem{corollary}[theorem]{Corollary}

\newenvironment{proofsketch}{\noindent{\textbf{Proof ($\star$) }}}{\hfill \qed}
\newenvironment{proofthm}{\noindent{\textbf{Proof of Theorem~\ref{thm:vc-lower-bound}. }}}{\hfill \qed}
\newtheorem{reduction rule}{Reduction Rule}

\DeclareMathOperator{\VCD}{VCD}
\DeclareMathOperator{\NCTD}{NCTD}
\DeclareMathOperator{\NCTDp}{NCTD^+}
\DeclareMathOperator{\NCTM}{NCTM}
\DeclareMathOperator{\NCTMp}{NCTM^+}
\DeclareMathOperator{\B}{\mathcal{B}}
\DeclareMathOperator{\CC}{\mathcal{C}}
\DeclareMathOperator{\T}{\mathfrak{T}}
\DeclareMathOperator{\III}{\mathfrak{I}}
\DeclareMathOperator{\CCC}{\mathfrak{C}}
\DeclareMathOperator{\KKK}{\mathfrak{K}}
\DeclareMathOperator{\HHH}{\mathfrak{H}}
\DeclareMathOperator{\diam}{diam}
\newcommand{\covectors}{\ensuremath{\mathcal{L}}}

\newcommand{\srmfull}{\textsc{B-NCTD$^+$}\xspace}
\newcommand{\rmfull}{\textsc{B-NCTD}\xspace}
\newcommand{\setcover}{\textsc{Set Cover}\xspace}
\newcommand{\partsat}{\textsc{3-Partitioned-3-SAT}\xspace}
\newcommand{\vc}{\mathtt{vc}}
\newcommand{\tw}{\mathtt{tw}}
\newcommand{\ETH}{\textsf{ETH}}
\newcommand{\setrep}{\textsf{set-rep}}

\newcommand{\defproblem}[3]{%
      \noindent
      \begin{center}
      \begin{tcolorbox}[title=#1,left=0mm,top=0mm,bottom=0mm,right=0mm,boxsep=0.5mm]
        \begin{tabular}{p{.13\textwidth}p{.8\textwidth}}
            \textbf{Input:} & \parbox[t]{.8\textwidth}{#2}\\
            \addlinespace
            \textbf{Question:} & \parbox[t]{.8\textwidth}{#3}\\
        \end{tabular}
      \end{tcolorbox}
      \end{center}
      }
      
\newenvironment{keywords}{
       \list{}{\advance\topsep by0.35cm\relax\small
       \leftmargin=1cm
       \labelwidth=0.35cm
       \listparindent=0.35cm
       \itemindent\listparindent
       \rightmargin\leftmargin}\item[\hskip\labelsep
                                     \bfseries Keywords:]}
     {\endlist}
      
\begin{document}

\title{Non-Clashing Teaching Maps for Balls in Graphs\thanks{This is a preprint of~\cite{CCMR24}, a work published in the proceedings of COLT 2024.}}

\author[1]{Jérémie~Chalopin}

\author[1]{Victor~Chepoi}

\author[2]{Fionn~{Mc~Inerney}}

\author[1]{Sébastien~Ratel}

\affil[1]{{\small Aix-Marseille Universit\'e, Universit\'e de Toulon, CNRS, LIS, Marseille, France}}

\affil[2]{{\small Algorithms and Complexity Group, TU Wien, Vienna, Austria}}

\date{}

\maketitle

\begin{abstract}%
Recently, Kirkpatrick et al.~[ALT 2019] and  Fallat et al.~[JMLR 2023] introduced non-clashing teaching and showed it
to be the most efficient machine teaching model satisfying the benchmark for collusion-avoidance set by Goldman and Mathias.
A teaching map $T$ for a concept class $\mathcal{C}$ assigns a
(teaching) set $T(C)$ of examples to each concept
$C \in \mathcal{C}$. A teaching map is non-clashing if no pair of
concepts are consistent with the union of their
teaching sets. The size of a non-clashing teaching map (NCTM) $T$ is
the maximum size of a teaching set $T(C)$, $C \in \mathcal{C}$.
The non-clashing teaching dimension NCTD$(\mathcal{C})$ of
$\mathcal{C}$ is the minimum size of an NCTM for $\mathcal{C}$.
NCTM$^+$ and NCTD$^+(\mathcal{C})$ are defined analogously,
except the teacher may only use positive examples.

We study NCTMs and NCTM$^+$s for the concept class $\mathcal{B}(G)$ consisting of all balls of a graph $G$.
We show that the associated decision problem \textsc{B-NCTD$^+$} for NCTD$^+$ is \NP-complete in split, co-bipartite, and bipartite graphs.
Surprisingly, we even prove that, unless the \ETH\ fails, \textsc{B-NCTD$^+$} does not admit an algorithm running in time $2^{2^{o(\mathtt{vc})}}\cdot n^{\mathcal{O}(1)}$, nor a kernelization algorithm outputting a kernel with $2^{o(\mathtt{vc})}$ vertices, where $\mathtt{vc}$ is the vertex cover number of $G$.
We complement these lower bounds with matching upper bounds.
These are extremely rare results: it is only the second problem in \NP\ to admit such a tight double-exponential lower bound parameterized by $\mathtt{vc}$, and only one of very few problems to admit such an \ETH-based conditional lower bound on the number of vertices in a kernel.
For trees, interval graphs, cycles, and trees of cycles, we derive NCTM$^+$s or NCTMs for $\mathcal{B}(G)$ of size proportional to its VC-dimension.
For Gromov-hyperbolic graphs, we design an approximate NCTM$^+$ for $\mathcal{B}(G)$ of size $2$, in which only pairs of balls with Hausdorff distance larger than some constant must satisfy the non-clashing condition.
\end{abstract}

{}

\section{Introduction}

Machine teaching is a core paradigm in computational learning theory that has attracted significant attention due to its applications in diverse areas
such as trustworthy AI~\cite{MZ15,ZZW18}, inverse reinforcement learning~\cite{BN19,HLMCA16}, robotics~\cite{ACYT12,TC09}, and education~\cite{Campy18,Zhu15}
(see~\cite{ZSZR18} for an overview).
In machine teaching models, given a concept class $\mathcal{C}$, a teacher presents to a learner a carefully chosen set $T(C)$
of correctly labeled examples from a concept $C\in \CC$ in such a way that the learner can reconstruct $C$ from $T(C)$. This defines the \emph{teaching map (TM)} $T$ and
the \emph{teaching sets} $T(C), C\in \CC$. The goal is to find a TM that minimizes the size of a largest teaching set.
The examples selected in $T(C)$ by the teacher are the most useful to the learner to reconstruct $C$, in contrast to models of
learning (like the classical PAC-learning) where the learner must reconstruct a concept of $\CC$ from randomly chosen examples.

There are a multitude of formal models of machine teaching~\cite{Balbach,GRSZ16,GCS17,GK95,GM96,MCV19,SM91,ZLH11}, which differ by the conditions imposed
on the teacher and learner. Several of these are \emph{batch teaching models}, where the examples proposed by the teacher to the learner are sets.
This is in contrast to \emph{sequential teaching models}, where the examples are not presented all at once, but rather in an order chosen by the teacher.
A key notion in formal models of machine teaching is that the teacher and learner should not \emph{collude}.
The benchmark for preventing this is the \emph{Goldman-Mathias (GM) collusion-avoidance criterion}~\cite{GM96}, which essentially demands a teaching map $T$ to admit a learner that returns the concept $C$ whenever it is shown any set of labeled examples that include $T(C)$ and are consistent with $C$. Recently, a batch teaching model called \emph{non-clashing teaching} (\emph{NC-teaching})
was proposed~\cite{FKS23,KSZ19}.
Given a concept class $\CC$, a TM $T$ on $\CC$ is \emph{non-clashing}  if, for any two distinct concepts $C,C'\in \CC$, either $T(C)$ is
not consistent with $C'$ or $T(C')$ is not consistent with $C$, or both. They proved NC-teaching to be the most efficient model (in terms of the worst-case number of examples required) satisfying
the GM collusion-avoidance criterion. It is common to restrict the teacher to only presenting positive examples (see~\cite{Angluinjcss,Angluininf} for early successes of this approach). These models are vastly studied due to their pertinence in, \textit{e.g.}, grammatical inference~\cite{Denis01,SO94}, computational biology~\cite{WDM06,YJSS08}, and recommendation systems~\cite{SPK00}. For these reasons, \cite{FKS23,KSZ19}  also introduced and studied \emph{positive NC-teaching}, in which the teacher may only use positive examples.

As with PAC-learning, where the VC-dimension $\VCD(\CC)$ of  $\CC$ drives the number of randomly chosen examples that are sufficient to learn the concepts of $\CC$, various models of machine teaching lead to different notions of teaching dimension which bound the teaching set sizes. The definitive teaching dimension (DTD)~\cite{GK95,SM91} is a prototypical one.
A \emph{definitive teaching set (DTS)} of a concept $C\in \CC$ is a $\CC$-sample (a sample consistent with a concept of $\CC$)
for which $C$ is the only consistent concept in $\CC$.
$\text{DTD}(\CC)$ is the maximum size of a DTS over all $C\in \CC$.
Note also the important recursive teaching dimension (RTD)~\cite{ZLHZ08,ZLH11}.
NC-teaching and positive NC-teaching also have dimension parameters.
The size of a TM $T$ on $\CC$ is the maximum size of $T(C)$ over all $C\in \CC$. The \emph{non-clashing teaching dimension} $\NCTD(\CC)$ (\emph{positive non-clashing teaching dimension} $\NCTDp(\CC)$, resp.) is the
minimum size of a non-clashing TM (NCTM) for $\CC$ (positive NCTM ($\NCTMp$) for $\CC$, resp.)~\cite{FKS23,KSZ19}.
An important research direction for various notions of teaching dimension is their relationship with the VC-dimension.
For NC-teaching, the authors of \cite{FKS23,KSZ19} say that ``The most fundamental
open question resulting from our paper is probably whether NCTD is upper-bounded by VCD in general'', and this open question is also mentioned in~\cite{Si}.

NCTMs are signed versions of representation maps (RMs) for concept
classes. The authors of~\cite{KuWa} introduced RMs to design
unlabeled sample compression schemes (SCSs) for maximum concept
classes. They proved that the existence of RMs of size $d=\VCD(\CC)$
for a maximum class $\CC$ is equivalent to the existence of unlabeled
SCSs of size $d$ for $\CC$. This
equivalence was generalized to ample classes in~\cite{ChChMoWa}, and they also constructed
RMs of size $d$ for maximum classes. The main difference between NCTMs
and RMs is that RMs assign to each concept
$C \in \mathcal{C}$ a set shattered by $\mathcal{C}$.
The authors of~\cite{LiWa} introduced SCSs, which have
been vastly studied due to their importance in computational machine learning.

In this paper, we consider NCTMs and $\NCTMp$s for the concept
class $\B(G)$ consisting of \emph{all balls of a graph $G$}.
For several graph classes, we derive NCTMs for $\B(G)$ of size proportional to $\VCD(\B(G))$.
Further, we study the computational complexity of the following decision problems:

\defproblem{\textsc{NCTD for Balls in Graphs} (\rmfull)}{A graph $G$ on $n$ vertices and a positive integer $k$.}{Is $\NCTD(\B(G))\leq k$?}

\defproblem{\textsc{NCTD$^+$ for Balls in Graphs} (\srmfull)}{A graph $G$ on $n$ vertices and a positive integer $k$.}{Is $\NCTDp(\B(G))\leq k$?}

\paragraph{Motivation for balls in graphs.} The \emph{combinatorial} and \emph{geometric} aspects of balls in graphs motivate \rmfull\ and \srmfull.
The combinatorial one ensures that the study of balls in graphs is as general as that of arbitrary concept classes.  Notably, to any set-family $\mathcal{C}\subseteq 2^V$, one can
associate a set of balls of a graph $G$ as follows. $V(G)=V\cup\{x_C: C\in \mathcal{C}\}$, the vertices of $\{x_C: C\in \mathcal{C}\}$ form a clique,
and $x_C$ and $v\in V$ are adjacent if and only if $v\in C$. For any $C\in \mathcal{\CC}$, $B_1(x_C)=C\cup \{x_{C'}: C'\in \mathcal{C}\}$.
On the other hand, one may hope that for graphs $G$ with a rich metric
structure, the geometric structure of $\B(G)$ may allow to efficiently construct NCTMs, which our results confirm. Further, in light of the open question of~\cite{FKS23,KSZ19,Si}, $\B(G)$ may provide graph classes where $\NCTD(\B(G))>\VCD(\B(G))$, \textit{e.g.}, we prove that, for trees of cycles, $\NCTD(\B(G))\le 4$, while $\VCD(\B(G))\leq 3$. Trees of cycles are planar, and,
for planar graphs, $\VCD(\B(G))\le 4$~\cite{BouTh,ChEsVa}, but it is unclear how to bound $\NCTD(\B(G))$ by a small constant (it is at most $615$ as, for any set-family $\CC$ with $\VCD(\CC)=d$, $\NCTD(\CC)\leq \text{RTD}(\CC)\leq 39.3752d^2 - 3.633d$~\cite{FKS23,HWLW17}). Lastly, SCSs for balls in graphs were studied in~\cite{ChChMc}.


\paragraph{Our Results.}
Our focus is twofold: we show that 1) \srmfull is computationally hard and exhibits rare properties from parameterized complexity, and 2) for several  graph classes, we derive $\NCTM$s of size proportional to $\VCD(\B(G))$.
We begin with the first direction, proving that:

\begin{enumerate}[leftmargin=*,topsep=1pt]
\item \srmfull is \NP-complete in split and co-bipartite graphs with a universal vertex, and bipartite graphs of diameter $3$.
\end{enumerate}

Note that~\cite{KSZ19} proved that it is \NP-hard to decide, for a concept class $\CC$, whether $\NCTD(\CC)=k$ or $\NCTDp(\CC)=k$, even if $k=1$, but their results do not apply to $\B(G)$ as they rely on the fact that deciding whether $\NCTDp(\CC)=1$ is \NP-hard.\footnote{$\NCTDp(\B(G))=1$ if and only if $G$ is edgeless. Let $uv\in E(G)$. $T(B_1(v))$ must contain a vertex in $N(v)$ as $T(B_0(v))\subseteq \{v\}$, say $u\in T(B_1(v))$. Then, $T(B_1(v))$ must contain another vertex in $N[v]$ as $T(B_0(u))\subseteq \{u\}$.}

\srmfull\ being \NP-hard in these graph classes motivates studying its parameterized complexity, as was done for other problems in learning theory~\cite{BGS23,DF93,GK21,LL18}. This leads to our first \emph{main result}, which exhibits the extreme computational complexity of the problem. Recently, the authors of~\cite{FGKLMST} developed a technique to prove double-exponential dependence on the treewidth ($\tw$) and/or the vertex cover number ($\vc$) of the graph in the running time of \FPT\ algorithms for problems in~\NP. For these classic structural parameters, they proved that, unless the~\ETH\ fails,\footnote{Roughly, the Exponential Time Hypothesis (\ETH) states that $n$-variable $3$-SAT cannot be solved in time~$2^{o(n)}$.} even on bounded-diameter graphs \textsc{Metric Dimension} and \textsc{Geodetic Set} do not admit $2^{2^{o(\tw)}}\cdot n^{\mathcal{O}(1)}$-time algorithms, \textsc{Strong Metric Dimension} does not admit a $2^{2^{o(\vc)}}\cdot n^{\mathcal{O}(1)}$-time algorithm, and these bounds are tight. Notably, these were the \emph{only} problems in \NP\ known to admit such tight double-exponential lower bounds, \emph{until now}.\footnote{After a preprint of this paper appeared on arXiv, tight double-exponential dependence on the treewidth was also shown for the \NP-complete problems \textsc{Test Cover} and \textsc{Locating-Dominating Set}~\cite{chakraborty2024tight}.} Applying this technique, we obtain our first \emph{main result}:

\begin{enumerate}[leftmargin=*,topsep=1pt]
\setcounter{enumi}{1}
\item Unless the \ETH\ fails, \srmfull\ does not admit a $2^{2^{o(\vc)}}\cdot n^{\mathcal{O}(1)}$-time algorithm, even in diameter-$3$ graphs.
\end{enumerate}

Our lower bound is robust as all the traditional structural parameters like treewidth, pathwidth, and treedepth are at most $\vc+2$. Further, \srmfull\ is only the \emph{second} problem in \NP\ to admit a double-exponential dependence in $\vc$. Thus, \srmfull\ is also incredibly interesting from a purely theoretical perspective, and it is strongly linked with other metric graph problems. The same reduction yields two more results, the first of which is also extremely rare:

\begin{enumerate}[leftmargin=*,topsep=1pt]
\setcounter{enumi}{2}
\item Unless the \ETH\ fails, \srmfull\ does not admit a kernelization algorithm outputting a kernel with $2^{o(\vc)}$ vertices, even in diameter-$3$ graphs.
\item Unless the \ETH\ fails, \srmfull\ does not admit a $2^{o(n)}$-time algorithm, even in diameter-$3$ graphs.
\end{enumerate}

Indeed, such \ETH-based conditional lower bounds on the number of vertices in a kernel are very rare as they are only known for a few other problems~\cite{chakraborty2024tight,CIK16,CPP16,Foucaud24b,FGKLMST,pratik24}. We show that our lower bounds concerning $\vc$ are tight by giving matching upper bounds:

\begin{enumerate}[leftmargin=*,topsep=1pt]
\setcounter{enumi}{4}
\item \srmfull\ admits a $2^{2^{\mathcal{O}(\vc)}}\cdot n^{\mathcal{O}(1)}$-time algorithm.
\item \srmfull\ admits a kernelization algorithm outputting a kernel with $2^{\mathcal{O}(\vc)}$ vertices.
\end{enumerate}

For \srmfull, we also give a $2^{\mathcal{O}(n^2 \cdot \diam)}$-time algorithm, yielding a $2^{\mathcal{O}(n^2)}$-time algorithm in bounded-diameter graphs.
We then focus on our second goal: designing NCTMs for $\B(G)$ that are linear in $\VCD(\B(G))$, when $G$ is restricted to certain graph classes. Proving that any ball in a tree or interval graph can be distinctly represented by two of its ``farthest apart'' vertices, we show that:

\begin{enumerate}[leftmargin=*,topsep=1pt]
\setcounter{enumi}{6}
\item If $G$ is a tree or an interval graph, then $\NCTD(\B(G)) \leq \NCTDp(\B(G))=\VCD(\B(G))\leq 2$.
\end{enumerate}

In contrast to trees and interval graphs, we prove that:

\begin{enumerate}[leftmargin=*,topsep=1pt]
\setcounter{enumi}{7}
\item Cycles do not admit NCTM$^+$s of fixed size for $\B(G)$, but do admit NCTMs of size~$2$.
\end{enumerate}

With this in mind, we search for NCTMs for $\B(G)$ for richer graph classes.
This already proves difficult in trees of cycles, for which, by a technical proof, we get the following:

\begin{enumerate}[leftmargin=*,topsep=1pt]
\setcounter{enumi}{8}
\item If $G$ is a tree of cycles, then $\text{NCTD}(\B(G))\le 4$, while $\VCD(\B(G))\leq 3$.
\end{enumerate}

In analogy to PAC-learning, in approximate NCTM$^+$s, only pairs of balls with Hausdorff distance larger than some constant must satisfy the non-clashing condition. Akin to our method in trees, we show that:

\begin{enumerate}[leftmargin=*,topsep=1pt]
\setcounter{enumi}{9}
\item If $G$ is a $\delta$-hyperbolic graph, then $\B(G)$  admits a $2\delta$-approximate NCTM$^+$ of size~$2$.
\end{enumerate}

\section{Preliminaries}
This section consists of definitions and notation.
For the ball $B_r(x)$, $T(x,r)$ denotes $T(B_r(x))$.
\textit{Omitted proofs of theorems and sketches of proofs (marked by~$\star$) are in the appendix.}

\paragraph{Concept classes and samples.}
In machine learning, a \emph{concept class} on a set $V$ is any collection $\CC$ of subsets of $V$.
The \emph{VC-dimension} $\VCD(\mathcal{C})$ of $\mathcal{C}$ is the size of a largest set
$S\subseteq V$ shattered by $\mathcal{C}$, that is, such that $\{C\cap S: C\in \mathcal{C}\}=2^S$.
A \emph{sample} is a set $X=\{(x_1,y_1),\ldots,(x_m,y_m)\}$, where $x_i\in V$ and $y_i\in\{-1,+1\}$.
A sample $X$ is \emph{realizable by a concept} $C\in \CC$  if $y_i=+1$ when $x_i\in C$, and $y_i=-1$ when $x_i\notin C$.
A sample $X$ is a  \emph{$\CC$-sample} if $X$ is realizable  by some concept $C$ in $\CC$. To encode concepts of a concept class $\CC$ on $V$ and $\CC$-samples, we use the language of  sign vectors
from oriented matroids theory~\cite{BLSWZ}. Let $\covectors$ be a non-empty \textit{set of sign vectors}, \textit{i.e.},
maps from $V$ to $\{\pm 1,0\} := \{-1,0,+1\}$.
For $X \in \covectors$, let $X^+ := \{v\in V: X_v=+1\}$ and $X^-:= \{v\in V:
X_v=-1\}$. The set $\underline{X} := X^+ \cup X^-$
is called the \emph{support} of $X$.
We denote by $\preceq$  the product ordering
on $\{ \pm 1,0\}^V$ relative to the ordering of signs with $0 \preceq -1$  and $0 \preceq +1$.
Any $\CC\subseteq 2^V$ can be viewed as a set of sign vectors of $\{ \pm 1\}^V$:
each concept $C\in \CC$ is encoded by the sign vector $X(C)$, where $X_v(C)=+1$ if $v\in C$ and
$X_v(C)=-1$ if  $v\notin C$. In what follows, we consider $\CC$ simultaneously as a
collection of sets and as a set of  $\{ \pm 1\}$-vectors.
From the definition of a sample $X$, it follows that  $X$ is just a sign vector and that the samples realizable by a concept $C\in \CC$ are all $X\in \{ \pm 1,0\}^V$ such that $X\preceq C$.

\paragraph{NCTMs and NCTD.}
For a concept class $\CC$ on $V$, a TM $T$ associates, to each $C\in \CC$, a realizable sample $T(C)$ for $C$ (the \emph{teaching set} of $C$), \textit{i.e.}, $T(C)\in \{ \pm 1, 0\}^V$ and $T(C)\preceq C$.
Rephrasing the original definitions of NCTMs and RMs, 
a TM $T: \CC\rightarrow \{ \pm 1,0\}^V$ is \emph{non-clashing} if whenever $T(C')\preceq C$ and $T(C)\preceq C'$ for 
$C,C'\in \CC$, then $C=C'$. Equivalently, $T$ is non-clashing
if, for any two distinct concepts $C,C'$ of $\CC$, the \emph{non-clashing condition} holds: for all
$C,C'\in \CC$ with $C\ne C'$, $C|_{(\underline{T}(C)\cup \underline{T}(C'))}\ne C'|_{(\underline{T}(C)\cup \underline{T}(C'))}$.
If $T$ also satisfies
the \emph{inclusion condition}: $\underline{T}(C)=T^+(C)\subseteq C^+$ for any $C\in \mathcal{C}$,
then $T$ is an $\NCTMp$.
The \emph{size} of a TM $T$ for $\mathcal{C}$ is $\max\{ |\underline{T}(C)|: C\in \mathcal{C}\}$.
$\text{NCTD}(\CC)$ ($\text{NCTD}^+(\CC)$, resp.) is the
minimum size of an NCTM for $\CC$ ($\NCTMp$ for $\CC$, resp.).

\paragraph{Graphs.} In this paper, graphs are simple, connected,
and undirected, and logarithms are to the base 2. For a positive
integer $k$, $[k]:=\{1,\ldots,k\}$.  Given a graph $G$, its vertex set is $V(G)$ and its edge set is $E(G)$.
The \emph{distance} $d_G(u,v)$ between two vertices $u,v\in V(G)$ is
the length of a shortest $(u,v)$-path in $G$.  For any
$r\in \mathbb{N}$ and $u\in V(G)$, the \emph{ball of radius $r$
  centered at $u$} is $B_r(u):=\{v: d_G(u,v)\leq r\}$.  For
any $u\in V(G)$, $B_1(u)=N_G[u]$ and $N_G(u):=N_G[u]\setminus \{u\}$.
Two balls are \emph{distinct} if they are distinct as sets.
$\B(G)$ is the set of all distinct balls of $G$.
For $S\subseteq V(G)$, the \emph{diameter} of $S$ is
$\diam(S):=\max_{u,v\in S}d_G(u,v)$, and a \emph{diametral pair} of
$S$ is a pair $u,v\in S$ such that $d_G(u,v)=\diam(S)$.
The \emph{diameter} of $G$ is $\diam(G):=\diam(V(G))$.
For any $u,v\in V(G)$, the \emph{interval} $I_G(u,v)$
between $u$ and $v$ in $G$ is the set of vertices on a
shortest path between $u$ and $v$ in $G$.
The \emph{vertex cover number} $\vc(G)$ of $G$ is the minimum number
of vertices that are incident to all edges of $G$.
When the context is clear, $G$ is omitted from some of these
notations.

\paragraph{Parameterized complexity.}
An instance of a parameterized problem $\pi$ consists of an input $I$ of the non-parameterized problem and a parameter $k\in \mathbb{N}$.
A \emph{kernelization} algorithm for $\pi$ transforms, in polynomial time, an instance $(I,k)$ of $\pi$ into an \emph{equivalent} instance $(I',k')$ of $\pi$ with $|I'|,k'\leq f(k)$, for a computable function $f$.
A reduction rule is \emph{safe} if the input instance is a YES-instance if and only if the output instance is a YES-instance.
See~\cite{parambook15} for a book on the topic.

\paragraph{Examples.} To illustrate the notion, we present NCTMs for some concept classes in graphs.

\begin{example}
  Let $G$ be any graph of diameter 2 such that, for each edge $xy \in E(G)$, $B_1(x) \cup B_1(y) = V(G)$. Examples are the complete bipartite graph $K_{n,m}$ and the \emph{$n$-octahedron}, which is the complete graph $K_{2n}$ on $2n$ vertices minus a perfect matching. We define an NCTM $T$ for $\B(G)$ of size $2$ as follows:
  \begin{itemize}
  \item for all $x \in V(G)$, set $T^+(x,0) := \{x\}$ and
    $T^-(x,0):=\{y\}$ for some neighbor $y$ of $x$;
  \item for all $x \in V(G)$ such that
    $V(G) \setminus B_1(x) \neq \varnothing$, set $T^+(x,1):=\{ x\}$
    and $T^-(x,1):=\{z\}$ for some vertex $z$ at distance 2 from $x$. Also set $T^+(V(G)):=\varnothing$ and $T^-(V(G)):=\varnothing$.
  \end{itemize}
  We show that $T$ is non-clashing. For any $x\in V(G)$
  and $r\in \{0,1\}$, if $B_r(x) \neq V(G)$, then $T^-(x,r)\neq \varnothing$,
  and thus, $T$ is non-clashing for $B_r(x)$ and $V(G)$.
  Consider a ball $B_0(x)$ and let $\underline{T}(x,0) = \{x,y\}$. For
  any ball $B' \neq B_0(x)$ such that $x \in B'$ and $y \notin B'$,
  we have that $B'=B_1(y')$ for some neighbor $y'$ of $x$ distinct from
  $y$. Since $y' \in T^+(y',1) \setminus B_0(x)$, $T$ is non-clashing
  for $B_0(x)$ and any other ball. Consider now two balls $B_1(x)$ and
  $B_1(y)$ such that $B_1(x) \neq V(G)$ and $B_1(y) \neq V(G)$. If
  $d(x,y) = 2$, then $x \in T^+(x,1) \setminus B_1(y)$ and $T$ is
  non-clashing for $B_1(x)$ and $B_1(y)$. Suppose now that $x$ and $y$
  are adjacent and let $\underline{T}(x,1) = \{x,z\}$. Then, $z$ is
  adjacent to $y$, and thus, $T$ is non-clashing for $B_1(x)$ and
  $B_1(y)$.
\end{example}

There is no $\NCTMp$ of constant size for the example above, even for the $n$-dimensional octahedron $G$.
Indeed, for any $x \in V(G)$, there is a
unique $\bar{x} \in V(G) \setminus B_1(x)$. Thus, in order to distinguish
$B_1(x) = V(G) \setminus \{\bar{x}\}$ and $V(G)$, we must have $\bar{x}
\in T^+(V(G))$, and thus, $|T^+(V(G))| = |V(G)|$.

Our second example is the concept class $\CCC_5$ that does not come from the family of balls of a graph. It was given in~\cite{PaTa} as an example of a concept class with VC-dimension~$2$ that does not admit an unlabeled sample compression scheme of size 2.

\begin{example} The ground set of $\CCC_5$ is the vertex set of the 5-cycle $C_5$ (which we will suppose to be oriented counterclockwise). The concepts of $\CCC_5$ are of two types: the sets $\{ u,v\}$ of size 2 such that $u$ and $v$ are not adjacent in $C_5$, and the sets $\{ x,y,z\}$ defining a path of length 2 in $C_5$. If $C= \{ x,y,z\}$ is a path of length 2 in $C_5$ ordered counterclockwise, we call $y$ the \emph{middle vertex} of $C$, and $xy$ the \emph{first edge} of $C$. Consider the following $\NCTMp$ $T$ of size 2: for any concept $C=\{ u,v\}$ of size 2, set $T(C):=C=\{ u,v\}$; for any concept $C=\{ x,y,z\}$ of size 3, set $T(C):=\{ x,y\}$.

We show that $T$ is non-clashing. Any two concepts $C,C'$ of size 2 are clearly distinguished by their traces on $T(C)\cup T(C')$. Analogously, any two concepts $C,C'$ of size 3 have different first edges, and thus, are distinguished by their traces on $T(C)\cup T(C')$. Finally, any concepts $C=\{ u,v\}$ of size 2 and $C'=\{ x,y,z\}$ of size 3 are distinguished by their traces on $T(C)\cup T(C')$ since $T(C)=C=\{ u,v\}\ne \{ x,y\}=T(C')$ as $u$ and $v$ are not adjacent in $C_5$, while $x$ and $y$ are.
\end{example}

\section{\srmfull\ is \NP-complete}
In this section, we prove that \srmfull\ is \NP-complete for split, co-bipartite, and bipartite graphs. We reduce from the well-known \NP-hard \setcover\ problem defined as follows: given a set of elements $X=\{1,\ldots,n\}$, a family $\mathcal{S}=\{S_1,\ldots,S_m\}$ of subsets of $X$ that covers $X$ (\textit{i.e.}, whose union is $X$), and a positive integer $k$, does there exist $\mathcal{S}'\subset \mathcal{S}$ such that $\mathcal{S}'$ covers $X$ and $|\mathcal{S}'|\leq k$?

\begin{restatable}{theorem}{thmsplit}\label{thm:split}
\srmfull\ is \NP-complete in split and co-bipartite graphs with a universal vertex, and bipartite graphs of diameter 3.
\end{restatable}

\begin{proofsketch}
The problem is in \NP\ as any $\NCTMp$ for $\B(G)$ has a set of at most $n^2$ distinct balls as a domain.
We sketch the proof for split graphs. The proof for co-bipartite graphs is similar, while the one for bipartite graphs is more involved.

Let $\phi$ be an instance of \setcover\ in which each element of $X$ is in at most $m-2$ sets of $\mathcal{S}$.
From $\phi$, we construct the graph $G$ as follows.
Add the sets of vertices $V=\{v_1,\ldots,v_n\}$, $S=\{s_1,\ldots,s_m\}$, $U=\{u_1,\ldots,u_{m+1}\}$, and $W=\{w_1,\ldots,w_m\}$. For all $i\in [n]$ and $j\in [m]$, if $i\notin S_j$ in $\phi$, then add the edge $v_is_j$. For all $j,\ell \in [m]$ with $j\neq \ell$, add the edges $u_jw_{\ell}$ and $u_{m+1}w_{\ell}$. Make each vertex in $U$ adjacent to each vertex in $S$. Make each vertex in $V$ adjacent to each vertex in $W$. Lastly, make the vertices in $U \cup V$ form a clique.
We prove that $\phi$ admits a set cover of size at most $t$ if and only if there is an $\NCTMp$ of size at most $k=m+t$ for $\B(G)$.

Suppose that $\phi$ admits a set cover $\mathcal{S}'\subset \mathcal{S}$ of size at most $t$. Let $S'\subset S$ be such that, for all $j\in [m]$, $s_j\in S'$ if and only if  $S_j\in \mathcal{S}'$ in $\phi$. We define an $\NCTMp$ $T$ of size at most $k$ for $\B(G)$. Note that we only need to define $T$ for balls of $G$ of radius $0$ or $1$ as, for all $x\in V(G)$, $B_2(x)=B_1(u_{m+1})=V(G)$.
For all $x\in V(G)$, set $T(x,0):=\{x\}$.
For all $x\in V$, set $T(x,1):=B_1(x)\cap S$.
For all $x\in W \cup S$, set $T(x,1):=\{x,u_{m+1}\}$.  For
all $x\in U$, set $T(x,1):=S'\cup (B_1(x)\cap W)$.  Clearly, $T$ has size at most $k$ and satisfies the inclusion
condition. One can also check that $T$ is non-clashing. Thus, $T$ is an $\NCTMp$
of size at most $k$ for $\B(G)$.

Now, suppose that $\phi$ does not admit a set cover of size at most $t$.
For all $i\in [n]$ and $j\in [m]$, $B_1(u_{m+1})=B_1(u_j)\cup \{w_j\}$, $B_1(v_i)\subset B_1(u_{m+1})$, and $(B_1(u_{m+1})\setminus B_1(v_i))\subset S$. Hence, for any $\NCTMp$ $T$ for $\B(G)$, $W\subseteq T(u_{m+1},1)$ and $T(u_{m+1},1)\cap S$ corresponds to a set cover.
Consequently, $|T(u_{m+1},1)\cap S|>t$, and thus, $|T(u_{m+1},1)|>k$ for any $\NCTMp$ $T$ for $\B(G)$.
\end{proofsketch}

\section{Tight bounds for parameterizations by the vertex cover number}
In this section, we consider \srmfull\ parameterized by the vertex cover number $\vc$ of $G$.
For any $x\in V(G)$ and $r\in \mathbb{N}$, there are at most $2^n$ possibilities for $T(x,r)$, and there are at most $n \cdot \diam$ unique balls in $G$ (as it is connected). Thus, we get the following algorithm that will be needed later.

\begin{restatable}{proposition*}{propexact}\label{exact-algo}
\srmfull\ and \rmfull\ admit algorithms running in time $2^{\mathcal{O}(n^2 \cdot \diam(G))}$.
\end{restatable}

We use a recently introduced technique from~\cite{FGKLMST} to prove the following:

\begin{theorem}\label{thm:vc-lower-bound}
Unless the \ETH\ fails, even in graphs of diameter~$3$, \srmfull\ does not admit
\begin{itemize}
\item an algorithm running in time $2^{2^{o(\vc)}}\cdot n^{\mathcal{O}(1)}$, nor
\item a kernelization algorithm outputting a kernel with $2^{o(\vc)}$ vertices, nor
\item an algorithm running in time $2^{o(n)}$.
\end{itemize}
\end{theorem}

Using this technique, we prove Theorem~\ref{thm:vc-lower-bound} via a reduction from \partsat, introduced in~\cite{LaMeVa} and defined as follows.
Given a $3$-CNF formula $\phi$ and a partition of its variables into three disjoint sets $X^{\alpha}$, $X^{\beta}$, $X^{\gamma}$ such that $|X^{\alpha}|=|X^{\beta}|=|X^{\gamma}|=N$ and no clause contains more than one variable from any of $X^{\alpha}$, $X^{\beta}$, and $X^{\gamma}$, is $\phi$ satisfiable?
The crux of the technique is to replace edges between clause and variable vertices by a ``small'' separator, called a set representation gadget, that encodes these relationships.
The proof of {\cite[Theorem 3]{LaMeVa}} along with the Sparsification Lemma~\cite{IPZ} implies the following:

\begin{proposition}\label{prop-eth}
Unless the \ETH\ fails, \partsat\ does not admit an algorithm running in time $2^{o(M)}$, where $M$ is the number of clauses.
\end{proposition}

\paragraph{Set representation gadget.}\label{sec:gadget}
Let $p$ be the smallest integer such that $3M\leq \binom{2p}{p}$, and observe that $p=\mathcal{O}(\log M)$.
Let $\mathcal{F}_p$ be the collection of subsets of $[2p]$ that contain exactly $p$ integers.
Define $\setrep:[3M] \rightarrow [\mathcal{F}_p]$ as a one-to-one
function by arbitrarily assigning a set in $\mathcal{F}_p$ to each
integer in $[3M]$.
Consider the variables from $X^{\alpha}$ in $\phi$.
For each variable $x^{\alpha}_i$ ($i\in [N]$) in $X^{\alpha}$ in $\phi$, there are two vertices $t^{\alpha}_{2i}$ and $f^{\alpha}_{2i-1}$ corresponding to the positive and negative literals of $x^{\alpha}_i$, respectively.
For each clause $C_j$ ($j\in [M]$) in $\phi$, there is a clause vertex $c_j$.
Add a set of vertices $V^{\alpha}=\{v^{\alpha}_1,\ldots,v^{\alpha}_{2p}\}$.
For each $i\in [N]$, add the edge $t^{\alpha}_{2i}v^{\alpha}_{p'}$ for each $p'\in \setrep(2i)$.
Similarly, for each $i\in [N]$, add the edge $f^{\alpha}_{2i-1}v^{\alpha}_{p'}$ for each $p'\in \setrep(2i-1)$.
Now, for all $i\in [N]$ and $j\in [M]$, if the variable $x^{\alpha}_i$ appears as a positive (negative, resp.) literal in the clause $C_j$ in $\phi$, then add the edge $c_jv^{\alpha}_{p'}$ for each $p'\in [2p]\setminus \setrep(2i)$ ($p'\in [2p]\setminus \setrep(2i-1)$, resp.).
For all $j\in [M]$, if no variable from $X^{\alpha}$ appears in $C_j$ in $\phi$, then make $c_j$ adjacent to all the vertices in $V^{\alpha}$.
See Fig.~\ref{fig:vc-figures} (right).
As a clause contains at most one variable from $X^{\alpha}$ in $\phi$, $t^{\alpha}_{2i}$ ($f^{\alpha}_{2i-1}$, resp.) and $c_j$ do not share a common neighbor in $V^{\alpha}$ if and only if the clause $C_j$ contains $x^{\alpha}_i$ as a positive (negative, resp.) literal in $\phi$.
We exploit this for the reduction, and since $p=\mathcal{O}(\log M)$, this ensures that $\vc(G)=\mathcal{O}(\log M)$.

\paragraph{Reduction.}\label{sec:red-vc}
Let $\phi$ be an instance of \partsat\ on $3N$ variables and $M=\mathcal{O}(N)$ clauses such that $M>N$. For all $\delta\in\{\alpha,\beta,\gamma\}$, let the variables in $X^{\delta}$ be $x^{\delta}_1,\ldots,x^{\delta}_N$. From $\phi$, we construct the graph $G$ as follows.
Add the vertex sets $C=\{c_1,\ldots,c_M\}$, $W=\{w_1,\ldots,w_{3M}\}$, and $U=\{u_1,\ldots,u_{3M}\}$, and the vertices $u_{3M+1}$, $u'_{3M+1}$, and $z$.
For all $\delta\in \{\alpha,\beta,\gamma\}$ and $i\in [N]$, add the vertices $t^{\delta}_{2i}$ and $f^{\delta}_{2i-1}$, and let $A^{\delta}=\{t^{\delta}_{2i}\mid i\in [N]\}\cup \{f^{\delta}_{2i-1}\mid i\in [N]\}$.
For all $\delta\in \{\alpha,\beta,\gamma\}$, add two independent sets of $2p$ vertices $V^{\delta}=\{v^{\delta}_1,\ldots,v^{\delta}_{2p}\}$ and $V^{\delta,*}=\{v^{\delta,*}_1,\ldots,v^{\delta,*}_{2p}\}$, and make all of them adjacent to each vertex in $U$. Also, add a clique of $2p$ vertices $V^W=\{v^W_1,\ldots,v^W_{2p}\}$.
For all $i\in [N]$ and $\delta\in \{\alpha,\beta,\gamma\}$, add the edge $t^{\delta}_{2i}v^{\delta}_{p'}$ ($f^{\delta}_{2i-1}v^{\delta}_{p'}$, resp.) for all $p'\in \setrep(2i)$ ($p'\in \setrep(2i-1)$, resp.).
For all $i\in [N]$, $j\in [M]$, and $\delta\in \{\alpha,\beta,\gamma\}$, if the variable $x^{\delta}_i$ appears as a positive (negative, resp.) literal in the clause $C_j$ in $\phi$, then add the edge $c_jv^{\delta}_{p'}$ for all $p'\in [2p]\setminus \setrep(2i)$ ($p'\in [2p]\setminus \setrep(2i-1)$, resp.).
For all $j\in [M]$ and $\delta\in \{\alpha,\beta,\gamma\}$, if no variable from $X^{\delta}$ appears in $C_j$ in $\phi$, then make $c_j$ adjacent to each vertex in $V^{\delta}$.
For all $\ell\in [3M]$, add the edge $w_{\ell}v^W_{p'}$ ($u_{\ell}v^W_{p'}$, resp.) for all $p'\in \setrep(\ell)$ ($p'\in [2p]\setminus \setrep(\ell)$, resp.).

Now, to ensure that, for all $i\in [N]$ and $\delta\in \{\alpha,\beta,\gamma\}$, exactly one of $t^{\delta}_{2i}$ and $f^{\delta}_{2i-1}$ is in $T(V(G))$ (each variable is assigned exactly one truth value), do the following:
(i) Add a clause vertex $c^{\delta}_i$, and let $C^{\delta}=\{c^{\delta}_i \mid i\in [N]\}$.
(ii) Add the edges $t^{\delta}_{2i}v^{\delta,*}_{p'}$ and $f^{\delta}_{2i-1}v^{\delta,*}_{p'}$ for all $p'\in \setrep(2i)$. Then, $t^{\delta}_{2i}$ and $f^{\delta}_{2i-1}$ have the same neighbors in $V^{\delta,*}$.
(iii) Add the edge $c^{\delta}_iv^{\delta,*}_{p'}$ for all $p'\in [2p]\setminus \setrep(2i)$. Note that $c^{\delta}_i$ and $t^{\delta}_{2i}$ ($f^{\delta}_{2i-1}$, resp.) have no common neighbors in $V^{\delta,*}$. One can think of $c^{\delta}_i$ as a clause $(x^{\delta}_i \lor \overline{x}^{\delta}_i)$.
(iv) Make $c^{\delta}_i$ adjacent to each vertex in $V^{\delta',*}$ for all $\delta'\in \{\alpha,\beta,\gamma\}$ such that $\delta'\neq \delta$.

Lastly, make $z$ adjacent to each vertex in $U \cup A^{\alpha} \cup A^{\beta} \cup A^{\gamma} \cup \{u_{3M+1},u'_{3M+1}\}$, $u_{3M+1}$ adjacent to each vertex in $V(G)\setminus (U \cup A^{\alpha} \cup A^{\beta} \cup A^{\gamma})$, and $u'_{3M+1}$ adjacent to each neighbor of $u_{3M+1}$.
See Fig.~\ref{fig:vc-figures} (left). The reduction returns $(G,k)$ as an instance of \srmfull\ where $k=3N+3M$.

\begin{figure}[htb]
\centering
\includegraphics[scale=0.7]{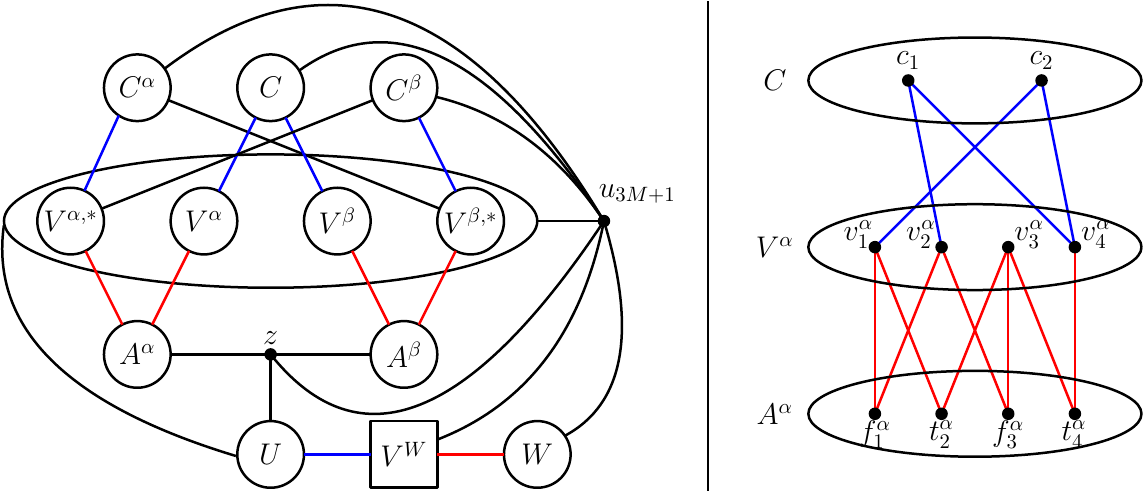}
\caption{
Graph $G$ in the proof of Thm.~\ref{thm:vc-lower-bound} (left) and its sets $A^{\alpha}$, $V^{\alpha}$, $C$ (right). For clarity, we omit $C^{\gamma},V^{\gamma},V^{\gamma,*},A^{\gamma}$, $u'_{3M+1}$. In $\phi$, $x^{\alpha}_1$ appears as a positive literal in $C_1$, and $x^{\alpha}_2$ as a negative literal in $C_2$. Red and blue edges are according to $\setrep$ (in a complementary~way).
}
\label{fig:vc-figures}
\end{figure}

\begin{restatable}{lemma*}{mainforward}\label{lem:vc-red-forward}
If $\phi$ is satisfiable, then $G$ admits an $\NCTMp$ for $\B(G)$ of size $k$.
\end{restatable}

\begin{proofsketch}
Let $\pi: X^{\alpha} \cup X^{\beta} \cup X^{\gamma} \rightarrow \{\text{True},\text{False}\}$ be a satisfying assignment for $\phi$.
We define the set $\pi'\subset A^{\alpha} \cup A^{\beta} \cup A^{\gamma}$ corresponding to $\pi$.
For all $\delta\in \{\alpha,\beta,\gamma\}$ and $x^{\delta}_i$ in $\phi$, if $\pi(x^{\delta}_i)=\text{True}$, then $t^{\delta}_{2i}\in \pi'$, and otherwise, $f^{\delta}_{2i-1}\in \pi'$.
So, $|\pi'|=3N$.
We define an $\NCTMp$ $T$ of size $k$ for $\B(G)$ as follows. We need not define $T$ for $B_2(z)$, $B_2(u'_{3M+1})$, and balls of radius at least~$3$ as, for all $x\in V(G)$, $B_3(x)=B_2(z)=B_2(u'_{3M+1})=B_2(u_{3M+1})=V(G)$.

For all $x\in V(G)$, set $T(x,0):=\{x\}$.
For all $\delta\in \{\alpha,\beta,\gamma\}$ and $x\in A^{\delta}$, set $T(x,1):=B_1(x)$ and $T(x,2):=\{u_1,t^{\delta'}_2,t^{\delta''}_2\} \cup B_1(x)$, where $\{\delta,\delta',\delta''\} = \{\alpha,\beta,\gamma\}$.
For all $\delta\in \{\alpha,\beta,\gamma\}$ and $x\in V^{\delta} \cup V^{\delta,*}$, set $T(x,1):=B_1(x)\setminus U$ and $T(x,2):=\{w_1,z\} \cup (B_1(x)\setminus U)$.
For all $x\in V^W$, set $T(x,1):=\{u_{3M+1},u'_{3M+1}\} \cup V^W \cup (B_1(x)\cap U)$ and $T(x,2):=\{u_{3M+1},u'_{3M+1},z\} \cup V^W \cup U$.
For all $\delta\in \{\alpha,\beta,\gamma\}$ and $x\in C \cup C^{\delta}$, set $T(x,1):=B_1(x)$ and $T(x,2):=\{x,w_1\} \cup (B_2(x)\cap (A^{\alpha} \cup A^{\beta} \cup A^{\gamma}))$.
For all $x\in U \cup \{u_{3M+1},u'_{3M+1}\}$, set $T(x,1):=B_1(x)\setminus (C \cup C^{\alpha} \cup C^{\beta} \cup C^{\gamma})$ and $T(x,2):=(B_2(x)\cap W) \cup \pi'$.
For all $x\in W$, set $T(x,1):=B_1(x)$ and $T(x,2):=\{x,z,u_{3M+1},u'_{3M+1}\} \cup (B_2(x)\cap U)$.
Set $T(z,1):=\{z,u_1\}$.
One can verify that $T$ has size at most $k$ and satisfies the inclusion condition for each ball in $\B(G)$ and the non-clashing condition for all pairs of balls in $\B(G)$. Hence, $T$ is an $\NCTMp$ of size at most $k$ for $\B(G)$.
\end{proofsketch}

\begin{restatable}{lemma*}{mainbackward}\label{lem:vc-red-backward}
If $G$ admits an $\NCTMp$ for $\B(G)$ of size $k$, then $\phi$ is satisfiable.
\end{restatable}

\begin{proofsketch}
Let $T$ be an $\NCTMp$ for $\B(G)$ of size $k$.
For all $i \in [N]$, $\ell \in [3M]$, and $\delta \in \{\alpha, \beta, \gamma\}$, as $B_2(u_{3M+1})=V(G)$, $B_2(u_{\ell})=V(G)\setminus \{w_{\ell}\}$, and $B_2(c^{\delta}_i)=V(G)\setminus \{t^{\delta}_{2i},f^{\delta}_{2i-1}\}$, we have $|T(u_{3M+1},2)\cap W|=3M$ and $|T(u_{3M+1},2)\cap \{t^{\delta}_{2i},f^{\delta}_{2i-1}\}|\geq 1$.
Since $k=3N+3M$, the latter inequality is an equality.
From $T(u_{3M+1},2)$, we extract an assignment $\pi: X^{\alpha} \cup X^{\beta} \cup X^{\gamma} \rightarrow \{\text{True},\text{False}\}$ for~$\phi$.
For all $i \in [N]$ and $\delta \in \{\alpha, \beta, \gamma\}$, if $T(u_{3M+1},2)\cap \{t^{\delta}_{2i},f^{\delta}_{2i-1}\} = \{t^{\delta}_{2i}\}$, then set $\pi(x^{\delta}_i)=\text{True}$, and otherwise, set $\pi(x^{\delta}_i)=\text{False}$.
Thus, $\pi$ assigns each variable in $\phi$ exactly one truth value.
It is not hard to see that $\pi$ is a satisfying assignment for $\phi$.
\end{proofsketch}

\begin{proofthm}
By Lemmas~\ref{lem:vc-red-forward}~and~\ref{lem:vc-red-backward}, there is a reduction that takes an instance of \partsat\ and returns an equivalent instance $(G,k)$ of \srmfull\ with $\diam(G)=3$. As $|V(G)|=\mathcal{O}(M)$, unless the \ETH\ fails, \srmfull\ does not admit a $2^{o(n)}$-time algorithm by Prop.~\ref{prop-eth}.
$S = \{u_{3M+1},u'_{3M+1},z\} \cup V^W \cup V^{\alpha} \cup V^{\beta} \cup V^{\gamma} \cup V^{\alpha,*} \cup V^{\beta,*} \cup V^{\gamma,*}$ is a vertex cover of $G$ of size $\mathcal{O}(\log M)$.
Hence, a $2^{2^{o(vc)}}\cdot n^{\mathcal{O}(1)}$-time algorithm for \srmfull\ would imply a $2^{o(M)}$-time algorithm for \partsat, contradicting the \ETH\ by Prop.~\ref{prop-eth}.
Toward a contradiction, suppose that \srmfull\ admits a kernelization algorithm outputting a kernel with $2^{o(\vc)}$ vertices. Consider the following algorithm for \srmfull. Given an instance of \partsat with $M$ clauses, it applies the reduction to obtain an equivalent instance $(G,k)$ of \srmfull\ with $\vc(G)=\mathcal{O}(\log M)$. Then, it applies the assumed kernelization algorithm on $(G,k)$, outputting a kernel with $2^{o(\vc)}$ vertices. Finally, it applies the algorithm from Prop.~\ref{exact-algo} on the kernel, which takes $2^{\mathcal{O}(({2^{o(\vc)}})^2 \cdot \diam)} = 2^{o(M)}$ time, contradicting the \ETH\ by Prop.~\ref{prop-eth}.
\end{proofthm}

We now show that our $\vc$ lower bounds are tight:

\begin{restatable}{theorem}{kernel}\label{thm:vc-upper-bound}
\srmfull\ admits
\begin{itemize}
\item an algorithm running in time $2^{2^{\mathcal{O}(\vc)}}\cdot n^{\mathcal{O}(1)}$, and
\item a kernelization algorithm outputting a kernel with $2^{\mathcal{O}(\vc)}$ vertices.
\end{itemize}
\end{restatable}

Given a graph $G$, two vertices $u,v\in V(G)$ are \emph{false twins} if $N(u)=N(v)$.
The following reduction rule is used to design the kernelization algorithm in Theorem~\ref{thm:vc-upper-bound}.

\begin{reduction rule}[RR1]\label{red-rule}
Given a graph $G$ and a set $X\subseteq V(G)$ such that $I:=V(G)\setminus X$ is an independent set (\textit{i.e.}, $X$ is a vertex cover of $G$), if there exist $2^{|X|}+2$ vertices in $I$ that are pairwise false twins, then delete one of them.
\end{reduction rule}

The idea for proving the forward direction of the next lemma is that, by the pigeonhole principle, for any set $S\subseteq I$ of $2^{|X|}+2$ false twins, there exist $x,y\in S$ such that $x\in T(x,1)$, $y\in T(y,1)$, and $T(x,1)\setminus \{x\} = T(y,1)\setminus \{y\}$. Then, for any teaching set containing $y$, we can replace $y$ by $x$ or another vertex of $S$ if $x$ is already in the teaching set.

\begin{restatable}{lemma*}{rulesafe}\label{red-rule-safe}
Reduction Rule~\ref{red-rule} is safe for \srmfull.
\end{restatable}

Theorem~\ref{thm:vc-upper-bound} follows from exhaustively applying RR1 for the kernelization algorithm, and using the algorithm from Prop.~\ref{exact-algo} on the resulting kernel for the other algorithm.

\section{NCTMs for classes of graphs}

In this section, we construct 
NCTMs for balls of several simple classes of graphs: trees, interval graphs, cycles, and trees of cycles.
We also design approximate NCTMs for balls in $\delta$-hyperbolic graphs.
In each of our NCTMs $T$, for any ball $B_r(x)$, $T^+(x,r)$ consists of two vertices of $B_r(x)$ that are ``farthest apart''.
In each case except for (trees of) cycles, we set $T(x,r)=T^+(x,r)$.

\subsection{Trees} 
For any ball $B_r(x)$ of a tree $\T$, define $T(x,r)$ as any diametral pair $\{u,v\}$ of $B_r(x)$. 

\begin{proposition} \label{rep-map-trees}
For a tree $\T$, $T$ is an  $\NCTMp$ for $\B(\T)$, \textit{i.e.},  $\NCTDp(\B(\T))=\VCD(\B(\T))\leq 2$.
\end{proposition}

\begin{proof}
  For any ball $B_r(x)$, $T(x,r)\subseteq B_r(x)$.
  So, as $|T(x,r)|=2$ for any ball $B_r(x)$ with $r\geq 1$, $T$ is non-clashing for any pair of balls that includes a ball of radius $0$.
  Now, suppose $B_{r_1}(x)\ne B_{r_2}(y)$ and assume that there exists $z \in B_{r_2}(y) \setminus B_{r_1}(x)$. Let $T(y,r_2) = \{u,v\}$. Then, $d(x,z) + d(u,v) \leq \max \{d(x,u) + d(v,z), d(x,v) + d(u,z) \}$, say $d(x,z) + d(u,v) \leq d(x,v) + d(u,z)$.
  Since $u,z \in B_{r_2}(y)$, $d(u,z) \leq d(u,v)$.
  Thus, $v \notin B_{r_1}(x)$ as $d(x,v) \geq d(x,z) > r_1$.
  So, $T$ is non-clashing.
\end{proof}

\subsection{Interval graphs}

We consider a representation of an interval graph  $\III$
by a set of segments $J_v$, $v\in V(\III)$, of $\mathbb{R}$ with pairwise distinct ends.
For any $u\in V(\III)$, its segment is denoted by $J_u=[s_u,e_u]$, where $s_u$ is the start of $J_u$,
and $e_u$ is the end of $J_u$, \textit{i.e.}, $s_u\leq e_u$.
We use the following property: 

\begin{lemma}{(Lemma~24, \cite{ChChMc})}\label{lem:intervalinclusion}
    If $u,v\in B_r(x)$, $s_u,s_z<s_v$, and $e_u<e_v,e_z$, then $z\in B_r(x)$.
\end{lemma}

For a subgraph $\III'$ of $\III$, $\{u,v\}$ is a \emph{farthest pair}
of $\III'$ if $u$ is the vertex in $\III'$ whose segment $J_u$ ends farthest to the
left, and $v$ is the vertex in $\III'$ whose segment $J_v$ begins farthest to the right,
\textit{i.e.}, for any $w\in V(\III')\setminus \{u,v\}$, we have $e_u<e_w$ and $s_w<s_v$.
Define the map $T$ on $\B(\III)$: for any ball $B_r(x)$ of $\III$, set $T(x,r)$ to be the
farthest pair $\{ u,v\}$ of $B_r(x)$ if $r\geq 1$, and set $T(x,0):=\{x\}$.

\begin{proposition} \label{rep-map-interval-graphs}
    For an interval graph $\III$, $T$ is an $\NCTMp$
    for $\B(\III)$, \textit{i.e.}, $\NCTDp(\B(\III))=\VCD(\B(\III))\leq 2$.
\end{proposition}

\begin{proof}
  For any ball $B_r(x)$, $T(x,r)\subseteq B_r(x)$.
   So, as $|T(x,r)|=2$ for any ball $B_r(x)$ with $r\geq 1$, $T$ is non-clashing for any pair of balls that includes a ball of radius $0$.
  Now, consider two balls $B_{r_1}(x)$ and
  $B_{r_2}(y)$ such that $T(x,r_1) = \{u,v\} \subseteq
  B_{r_2}(y)$. For any $z \in B_{r_1}(x)$, we have $e_z>e_{u}$ and
  $s_z<s_{v}$, and thus, $z \in B_{r_2}(y)$ by
  Lemma~\ref{lem:intervalinclusion}, establishing that
  $B_{r_1}(x) \subseteq B_{r_2}(y)$. Consequently, $T$ sastisfies the
  non-clashing condition.
  Finally, $\VCD(\B(\III))\leq 2$~\cite{DuHaVi}.
\end{proof}

\subsection{Cycles}

In contrast with trees and interval graphs, we prove that:

\begin{proposition} \label{cycles-unlabeled} In cycles, the family of balls do not admit $\NCTMp$s
of constant size.
\end{proposition}

\begin{proof}
  Consider the cycle $\CCC_n$ with $n=2k+2\ge 4$ and suppose that
  $\CCC_n$ admits an $\NCTMp$ $T$ of size at most $k$.  Each ball
  $B_{k}(x)$ contains all the vertices of $\CCC_n$ except the vertex
  $\overline{x}$ opposite to $x$ in $\CCC_n$.  Hence,
  $T(x,k)\subset \CCC_n\setminus \{ \overline{x}\}=B_{k}(x)$.  Since
  $|T(x,k)| \leq k$, $\overline{x} \in T(z,k)$ for at least $n-k-1$
  vertices $z\ne x$.
  Thus, $\sum_{x \in \CCC_n}|T(x,k)| \geq n(n-k-1)$. But, since $|T(x,k)|\le k$, this sum is at most $nk$.  Therefore,
  $nk\ge n(n-k-1)$, and thus,  $n\le 2k+1$, a contradiction.
\end{proof}

However, any cycle $\CCC=\CCC_n$ with $n \geq 6$ admits an NTCM of size $2<\VCD(\B(\CCC))=3$.
Suppose that $\CCC$ is oriented counterclockwise. Each ball $B_r(x)$ of $\CCC$ either coincides with
$\CCC$ or is an arc of $\CCC$. When $B_r(x)$ is an arc, we can speak about
the \emph{first} and \emph{last vertices of}  $B_r(x)$ in the counterclockwise order, and about
the \emph{first vertex outside} of $B_r(x)$ (this vertex is adjacent to the last vertex of $B_r(x)$).
%
The NCTM $T$ for $\B(\CCC)$ is as follows. If a ball $B_r(x)$ covers $\CCC$, then
$T(x,r)=\varnothing$. Otherwise, $T^+(x,r)$ is the first vertex of $B_r(x)$ and
$T^-(x,r)$ is the first vertex outside of $B_r(x)$.
%

\begin{proposition}\label{cycles-labeled}
  For a cycle $\CCC=\CCC_n$, $T$ is an NCTM for $\B(\CCC)$, and if
  $n \geq 6$, then~$\text{NCTD}(\B(\CCC))=2<\VCD(\B(\CCC))=3$.
\end{proposition}

\begin{proof} First, we prove that $T$ is an NCTM for $\B(\CCC)$. Let $B=B_r(x)$ and $B'=B_{r'}(y)$ be two different balls of $\CCC$.
The non-clashing condition is immediate if one of the balls coincides with $\CCC$. Therefore, suppose that both balls
are arcs. Let $T(x,r)=\{ u,w\}$, where $u$ is the first vertex of $B_r(x)$ and $w$ is the first vertex outside of $B_r(x)$.
Analogously, $T(y,r')=\{ u',w'\}$, where $u'$ is the first vertex of $B_{r'}(y)$ and $w'$ is the first vertex outside of $B_{r'}(y)$.
Also, let $v$ and $v'$ be the last vertices of $B_r(x)$ and $B_{r'}(y)$, respectively. Since $B_r(x)$ and $B_{r'}(y)$ are arcs of $\CCC$,
they are either (1) disjoint,  or (2) one is a proper subset of another, or (3) they overlap on a proper arc, or (4) they overlap on two arcs and together
cover $\CCC$. If $B_r(x)$ and $B_{r'}(y)$ are disjoint, then $u\in T^+(x,r)\setminus B_{r'}(y)$.
If $B_{r'}(y)\subseteq B_r(x)$, then $B_r(x)$ and $B_{r'}(y)$ differ with respect to at least at one of the vertices $u$ or $v$.
If $u\ne u'$, then $u\in T^+(x,r)\setminus B_{r'}(y)$.
If $u=u'$ and $v'\ne v$, then $w'\in B_r(x)$, while $w'\in T^-(y,r')$.
Now, suppose that $B_r(x)$ and $B_{r'}(y)$ overlap on a proper arc of each of $B_r(x)$ and $B_{r'}(y)$.
With respect to the counterclockwise order, this can be either the arc between $u'$ and $v$ or the arc between $u$ and $v'$, say the first (the other case is symmetric).
In this case, $u\in T^+(x,r)\setminus B_{r'}(y)$.
Finally, suppose that $B_r(x)$ and $B_{r'}(y)$ overlap on two arcs and cover $\CCC$.
These two arcs are defined by $u$ and $v'$ and by $u'$ and $v$ (in the counterclockwise order).
Then, $w\in B_{r'}(y)$, while $w\in T^-(x,r)$.
Thus, in all cases, $T$ satisfies the non-clashing condition.
Hence, $T$ is an NCTM for $\B(\CCC)$ of size 2.

To prove the lower bound, let $k = \lfloor n/2 \rfloor$ (\textit{i.e.},
$n \in \{2k,2k+1\}$) and assume that $\CCC$ admits an NCTM $T$ of
size $1$. There are $kn+1$ distinct balls in
$\CCC$ ($kn$ proper balls and $B_k(x)=\CCC$). Since at most
one ball can have an empty teaching set, there are $kn$ balls that
have teaching sets of size $1$. As there are $2n$ possible
teaching sets of size 1 (each vertex has sign $\pm 1$), by the pigeonhole principle, if $kn>2n$, then there exist
two different balls $B,B'$ with $T^+(B)=T^+(B')$ and $T^-(B)=T^-(B')$. But then $T$ does not satisfy the non-clashing condition for $B$ and $B'$, contrary to the assumption that
$T$ is an NCTM. Thus, for any $n \geq 6$, there is no NCTM of size $1$ for $\B(\CCC)$.
\end{proof}

\subsection{Trees of cycles (cacti)}
A \emph{tree of cycles} (or \emph{cactus}) is a graph $\KKK$ in which each
2-connected component is a cycle or an edge.
For a vertex $v$ of $\KKK$ that is not a cut vertex, let $C(v)$ be the
unique cycle containing $v$. If $v$ is a cut vertex, then
$C(v)=\{ v\}$.
For any vertices $u,v$ of $\KKK$, let $C(u,v)$ be the union of all cycles
and/or edges on the unique path of $B(\KKK)$ between $C(u)$ and $C(v)$. Note
that $C(u,v)$ is a path of cycles.
A set $S\subseteq V(G)$ is \emph{gated} if, for any $u\in V(G)$, there exists $u'\in S$ (the gate of $u$, with $u'=u$ if $u\in S$) such that $u'\in I(u,v)$ for any $v\in S$.
Given a triplet $x,u,v$ of vertices of  $G$, a vertex $y$ is an \emph{apex of $x$ with respect to $u$ and $v$} if $y\in I(x,u)\cap I(x,v)$ and
$I(x,y)$ is maximal with respect to the inclusion.  
One can easily show that for any triplet $x,u,v$ of a tree of cycles $\KKK$, there exists a unique apex of $x$ with respect to $u$ and $v$ 
and that any cycle  and any  path of cycles  of $\KKK$ are gated.
Let $\KKK$ be a tree of cycles and $\B$ its set of balls.
For $B\in \B$, let $[B]=\{ B_r(x): B_r(x)=B\}$.  We call a ball $B_r(x)$
in $[B]$ \emph{minimal} if it has a minimal radius among all balls in $[B]$.

\begin{lemma}\label{min-ball} If $\{ u,v\}$ is a diametral pair of $B_r(x)$, $x'$ is the apex of $x$ with respect to $u,v$ and $r'=r-d(x,x')$, then $B_{r'}(x')=B_r(x)$. In particular, if $B_r(x)$ is a minimal ball, then $x\in C(u,v)$.
\end{lemma}

\paragraph{The NCTM.}
We define a map $T$ for $\B$ as follows.  Let $B \in \B$ and $u,v$ be a diametral pair of
$B$. We set $T^+(B):=\{ u,v\}$ ($T^+(B):=B$ if $|B|=1$). To define $T^-(B)$, let
$B_r(x)$ be a minimal ball from $[B]$.  By Lemma~\ref{min-ball}, $x$
belongs to $C(u,v)$.  If $x$ is a cut vertex of $C(u,v)$, then we set
$T^-(B):=\varnothing$. Otherwise, let $C$ be the unique cycle of
$C(u,v)$ containing $x$, and let $u'$ and $v'$ be the respective gates
of $u$ and $v$ in $C$.  For any vertex $z$ of $\KKK$, we denote by
$z'$ its gate in $C$.
Consider the set $Z(x,u,v)=\{ z\in V(\KKK): z' \notin I(x,u')$
$\cup I(x,v') \mbox{ and } d(x,z)=r+1\}.$ If $Z(x,u,v)=\varnothing$, then we set $T^-(B)=\varnothing$.  If $Z(x,u,v)\ne \varnothing$,
let $Z^u(x,u,v)= \{s \in Z(x,u,v):u' \in I(x,s)\}$ and $Z^v(x,u,v) = \{t \in Z(x,u,v):v' \in I(x,t)\}$. If $Z^u(x,u,v)$
($Z^v(x,u,v)$, resp.) is not empty, pick $s \in Z^u(x,u,v)$ ($t \in Z^v(x,u,v)$, resp.) such that the distance $d(u',s')$ ($d(v',t')$, resp.) is maximized.
We set $T^-(B):=\{s,t\}$ if $s$ and $t$ exist, and $T^-(B):=\{s\}$ or $T^-(B):=\{ t\}$ if only one of $s$ and $t$ exists (see Fig.~\ref{fig:rep-map_cactii}).
Observe that if the cycle $C$ is not completely included in
$B_{r}(x)$, then $s$ and $t$ exist, they belong to
$C$ (possibly $s = t$), $s=s'$, $t=t'$, and
$d(x,s) = d(x,t) = r+1$.
\begin{figure}
    \centering
    \includegraphics[scale=0.46]{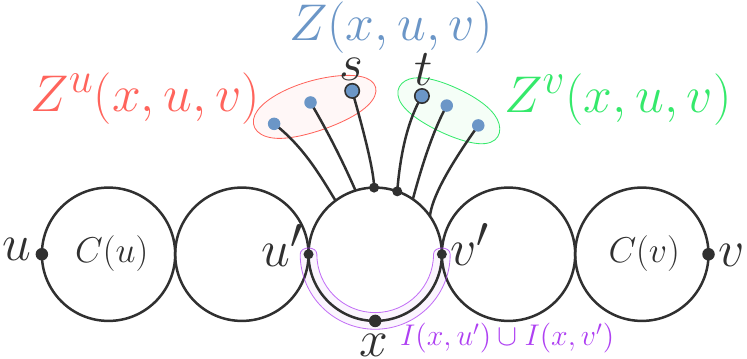}
    \caption{$Z(x,u,v)$, $Z^u(x,u,v)$, $Z^v(x,u,v)$, $s$, and $t$.
    }
    \label{fig:rep-map_cactii}
\end{figure}

\begin{restatable}{theorem}{treesofcycles} \label{treesofcycles} Let $B_1,B_2 \in \B$ be two balls of the same diameter in a tree of cycles $\KKK$ such that, for all $q \in \underline{T}(B_1)\cup \underline{T}(B_2)$, $q \in B_1$ if and only if $q \in B_2$. Then, $B_2 = B_1$.
Consequently, $T$ is an $\NCTM$
of size 4 for $\B(\KKK)$, and thus, $\text{NCTD}(\B(\KKK))\le 4$, while $\VCD(\B(\KKK)) \leq 3$.
\end{restatable}

\begin{proofsketch} First, analogous to the case of trees, $T$ is non-clashing for any pair of balls including a ball of radius $0$. Let $B_{r_1}(x)\in [B_1]$ be a minimal ball for  $B_1$ and $u,v$ be the diametral pair of $B_1$ defining $T^+(B_1)$.
  Analogously, let $B_{r_2}(y)\in [B_2]$ be a minimal ball for $B_2$. By contradiction, assume that $B_1\neq B_2$ and that
  $T$ does not satisfy the non-clashing condition for $B_1$ and $B_2$. Without loss of generality,
suppose that there exists a vertex $z \in B_2\setminus B_1$. One can show that $x$ cannot disconnect $z$ and $u$ (or $v$).
So, $x$ is not a cut vertex of $C(u,v)$. Hence, $x$ belongs to a
unique (gated) cycle $C$ of $C(u,v)$.
Since $\diam(B_1)=\diam(B_2)$ and $T^+(B_1)\subset B_2$, $u,v$ is also a diametral pair of $B_2$.

If $z' \in I(x,u')$, one can show that
$\diam(B_2) \geq d(v,z) > d(v,u) = \diam(B_1)$, a
contradiction. Thus, $z' \notin I(x,u')$ and, similarly,
$z' \notin I(x,v')$.  Since $x$ belongs to $C$, $I(x,u')$ and $I(x,v')$ intersect only in $x$, and their
union is the $(u',v')$-path passing via $x$.  By the previous
assertion, $z'$ belongs to the complementary $(u',v')$-path $P$ of $C$ and
$I(x,z')\cap \{ u',v'\}\ne \varnothing$. Hence, $Z(x,u,v)$ is
non-empty. Indeed, let $w$ be a vertex of $I(x,z)$ at distance $r_1+1$
from $x$. Then, either $z'$ is the gate of $w$ in $C$ or $w$ is a
vertex of the path $P$, and so, $w\in Z(x,u,v)$. Hence,
$Z(x,u,v)\ne\varnothing$, and thus, one of the vertices $s,t$
exists. If $s$ ($t$, resp.) exists, then its gate $s'$ ($t'$, resp.)
in $C$ belongs to $P$.

If $C \nsubseteq B_1$, then $s=s'$ and $t=t'$ disconnect $x$ and any
vertex of $P$. Since $x,z' \in B_2$ and $s,t \notin B_2$, necessarily
$z' \in I(u',s') \cup I(v',t')$. If $C\subseteq B_1$, then by the
definition of $s$ and $t$, we also have
$z' \in I(u',s') \cup I(v',t')$.
Assume, w.l.o.g., that $z' \in I(u',s')$ and recall that $z' \neq u'$.
Since $u'\in I(z',x)\subseteq I(s',x)$, $z'\in I(x,z)$, and $s'\in I(x,s)$,
we get $d(x,z)=d(x,u')+d(u',z')+d(z',z)$ and
$d(x,s)=d(x,u')+d(u',z')+d(z',s')+d(s',s)$. Since
$d(x,z)\geq r_1+1=d(x,s)$, we conclude that
$d(z',z)\geq d(z',s')+d(s',s)$.

By Lemma~\ref{min-ball} applied to $B_2=B_{r_2}(y)$, we have
$y\in C(u,v)$. Thus, $y=y'\in C$ or $y' \in
\{u',v'\}$. So, $d(y,z)=d(y,y')+ d(y',z')+d(z',z)$. As
$s\notin B_2$,
$r_2<d(y,s)=d(y,y')+d(y',s')+d(s',s) \leq
d(y,y')+d(y',z')+d(z',s')+d(s',s) \leq d(y,y')+d(y',z')+d(z',z)=d(y,z)
\leq r_2$, a contradiction. This gives $B_1=B_2$. The second assertion 
follows from the first one since two balls with
distinct diameters are distinguished by  diametral
pairs. Finally, $\VCD(\B(\KKK)) \leq 3$ since
trees of cycles $\KKK$ cannot be contracted to $K_4$, and if a graph $G$ does not contain $K_{d+1}$ as
a minor, then $\VCD(\B(G))\le d$~\cite{BouTh,ChEsVa}.
\end{proofsketch}

\subsection{Hyperbolic graphs}
Gromov's $\delta$-hyperbolicity is important
in metric geometry and geometric group theory, with
applications in analyzing real networks. It stipulates how
close, locally, the graph (or metric space) is to a tree metric.
For $\delta\ge 0$, a metric space $(X,d)$ is
\emph{$\delta$-hyperbolic}~\cite{Gr} if, for any four points $u,v,x,y$
of $X$,
$d(u,v)+d(x,y) \leq \max\{d(u,x)+d(v,y),d(u,y)+d(v,x)\}+2\delta$.
Two sets $A$ and $B$ of a metric space $(X,d)$ are
$\rho$-\textit{identical} if the \emph{Hausdorff distance} $d_H(A,B)$ between $A$ and $B$ is at most $\rho$. Otherwise,
they are \emph{$\rho$-distinct}.
For balls, $B_{r_1}(x)$ and $B_{r_2}(y)$ are
$\rho$-\textit{identical} if and only if $B_{r_1}(x)\subseteq B_{r_2+\rho}(y)$ and $B_{r_2}(y)\subseteq B_{r_1+\rho}(x)$.
A \emph{$\rho$-approximate $\NCTMp$} ($\text{NCTM}^+_\rho$) $T$
associates, to each ball $B_r(x)\in \B(G)$, a set
$T(x,r)\subseteq B_r(x)$ such that the non-clashing condition holds for
each pair of $\rho$-distinct balls.
For a $\delta$-hyperbolic graph $\HHH$ and any ball $B_r(x)$ of
$\HHH$, let $T(x,r)$ be any diametral pair of
$B_r(x)$ if $r\geq 1$, and set $T(x,0):=\{x\}$. Akin to our method in trees, we get that:

\begin{restatable}{theorem}{hyperbolic}\label{prop:LSCS_delta-hyperbolic}
  For a $\delta$-hyperbolic graph $\HHH$, $T$ is an
   $\text{NCTM}^+_{2\delta}$ of size $2$ for
  $\B(\HHH)$.
\end{restatable}

\section{Further work}
As is the case for general set-families, it would be interesting to know whether \srmfull\ and \rmfull\ are also para-\NP-hard parameterized by $k$. Further, it would be intriguing to know for which (other) structural parameterizations (\textit{e.g.}, treewidth, feedback vertex number, feedback edge number, and treedepth) they are tractable, and whether $\NCTD(\B(G))>\VCD(\B(G))$ for planar graphs. Lastly, as our $\NCTM$s are simpler than the SCSs in~\cite{ChChMc}, it seems reasonable to study them for notable and rich metric graph classes like median and Helly graphs.

\section*{Acknowledgments}
We thank the anonymous reviewers for their valuable feedback. This work has been supported by the Austrian Science Fund (FWF, project Y1329) and the ANR projects DISTANCIA (ANR-17-CE40-0015) and DUCAT (ANR-20-CE48-0006).

\newpage
\bibliography{rep-maps-bib}
\bibliographystyle{plainurl}

\appendix

\section{Full proofs of omitted or sketched proofs}

\thmsplit*

\begin{proof}
The problem is in \NP\ since any $\NCTMp$ $T$ for $\B(G)$ has a set of at most $n^2$ distinct balls as a domain,
and thus, it can be verified in polynomial time whether $T$ satisfies both the non-clashing condition for all pairs of balls in $\B(G)$, and the inclusion condition for each ball in $\B(G)$.
In all three cases, to prove that it is \NP-hard, we give a reduction from \setcover.

We begin with the case of split graphs. Let $\phi$ be an instance of \setcover\ with $X=\{ 1,\ldots,n\}$ and $\mathcal{S}=\{ S_1,\ldots,S_m\}$. We may also assume that $\phi$ is an instance in which each element of $X$ is contained in at most $m-2$ sets of $\mathcal{S}$. Indeed, any element contained in all of the sets of $\mathcal{S}$ will be covered by any choice of the sets of $\mathcal{S}$, and so, could simply be removed from the instance. In the resulting instance, for each of the elements contained in exactly $m-1$ sets of $\mathcal{S}$, it suffices to duplicate the set not containing that element to obtain the property that any element is contained in at most $m-2$ sets of $\mathcal{S}$. From this instance $\phi$, we construct the graph $G$ as follows. For all $i\in [n]$ and $j\in [m]$, add a vertex $v_i$ and a vertex $s_j$, and if $i\notin S_j$ in $\phi$, then add the edge $v_is_j$. Add the sets of vertices $U=\{u_1,\ldots,u_{m+1}\}$ and $W=\{w_1,\ldots,w_m\}$, and, for all $j,\ell \in [m]$ such that $j\neq \ell$, add the edge $u_jw_{\ell}$. For all $j\in [m]$, add the edge $u_{m+1}w_j$. Add edges so that every vertex in $U$ is adjacent to every vertex in $S=\{s_1,\ldots,s_m\}$. Add edges so that every vertex in $V=\{v_1,\ldots,v_n\}$ is adjacent to every vertex in $W$. Lastly, add edges so that the vertices in $U \cup V$ form a clique. This completes the construction of $G$, which is clearly achieved in polynomial time. See Figure~\ref{fig:split} for an illustration of $G$. Note that $u_{m+1}$ is a universal vertex, and that the vertices in $W \cup S$ form an independent set, and thus, $G$ is a split graph containing a universal vertex. We prove that $\phi$ admits a set cover of size at most $t$ if and only if there is an $\NCTMp$ of size at most $k=m+t$ for $\B(G)$.

\begin{figure}[htb]
\centering
\includegraphics[scale=0.75]{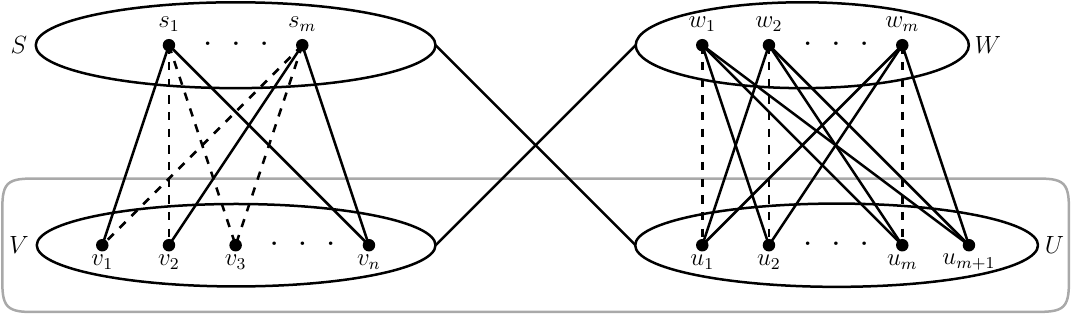}
\caption{\label{fig:split}
The split graph $G$ constructed in the proof of Theorem~\ref{thm:split}. Vertices contained in a rectangle form a clique. An edge between two ellipses indicates that each vertex in one ellipse is adjacent to each vertex in the other. Dashed lines highlight some non-existing edges. In this example, $2,3\in S_1$ (but $1,n\notin S_1$) and $1,3\in S_m$ (but $2,n\notin S_m$) in $\phi$.
}
\end{figure}

First, suppose that $\phi$ admits a set cover of size at most $t$, and let $\mathcal{S}'\subset \mathcal{S}$ be such a set cover. Let $S'\subset S$ be such that, for all $j\in [m]$, $S_j\in \mathcal{S}'$ in $\phi$ if and only if $s_j\in S'$. We define an $\NCTMp$ $T$ of size at most $k$ for $\B(G)$ as follows. Also, note that we only need to define $T$ for balls of $G$ of radius at most~$1$ since, for all $x\in V(G)$, $B_2(x)=B_1(u_{m+1})=V(G)$.

\begin{itemize}
\item For all $x\in V(G)$, set $T(x,0):=\{x\}$.
\item For all $x\in V$, set $T(x,1):=B_1(x)\cap S$ and note that $|T(x,1)|\geq 2$ since every element of $X$ is contained in at most $m-2$ sets of $\mathcal{S}$.
\item For all $x\in W \cup S$, set $T(x,1):=\{x,u_{m+1}\}$.
\item Finally, for all $x\in U$, set $T(x,1):=S'\cup (B_1(x)\cap W)$.
\end{itemize}

It is easy to verify that the map $T$ has size at most $k$ and satisfies the inclusion condition for each ball in $\B(G)$.
We now show that $T$ satisfies the non-clashing condition for all pairs of balls in $\B(G)$. For all $x\in V(G)$, $|B_0(x)|=|T(x,0)|=1$ and $|T(x,1)|\geq 2$, and thus, $T$ satisfies the non-clashing condition for all pairs of balls in $\B(G)$ where at least one of the balls has radius~$0$. For all $x \in W \cup S$, we have that $x\in T(x,1)$, and $W \cup S$ is an independent set. Hence, $T$ satisfies the non-clashing condition for all pairs of balls of radius~$1$ in $\B(G)$ centered at vertices in $W \cup S.$
For all $x,y\in V$, $B_1(x)\setminus S=B_1(y)\setminus S$, $T(x,1)\cap S=B_1(x)\cap S$, and $T(y,1)\cap S=B_1(y)\cap S$. Hence, $T$ satisfies the non-clashing condition for all pairs of balls of radius~$1$ in $\B(G)$ centered at vertices in $V.$
For all $x\in W \cup S$ and $y\in V$, $|B_1(x)\cap S|\leq 1$ and $|T(y,1)\cap S|\geq 2$. Hence, $T$ satisfies the non-clashing condition for all pairs of balls of radius~$1$ in $\B(G)$ centered at vertices in $W \cup S \cup V.$
For all $x,y\in U$, $B_1(x)\setminus W=B_1(y)\setminus W$, $T(x,1)\cap W=B_1(x)\cap W$, $T(y,1)\cap W=B_1(y)\cap W$. Hence, $T$ satisfies the non-clashing condition for all pairs of balls of radius~$1$ in $\B(G)$ centered at vertices in $U.$
For all $x\in U$ and $y\in W \cup S$, $|T(x,1)\cap W|\geq m-1$ and $B_1(y)\cap W\leq 1$. Hence, $T$ satisfies the non-clashing condition for all pairs of balls of radius~$1$ in $\B(G)$ centered at vertices in $W \cup S \cup U.$
For all $x \in V$, by the construction, we have that $B_1(x)\cap S'\neq S'$ since $S'$ corresponds to the set cover $\mathcal{S}'$, and, for all $y\in U$, $S'\subset T(y,1)$. Hence, $T$ satisfies the non-clashing condition for all pairs of balls of radius~$1$ in $\B(G)$ centered at vertices in $V \cup U.$
Combining all this, we get that $T$ is an $\NCTMp$ of size at most $k$ for $\B(G)$.

Now, suppose that $\phi$ does not admit a set cover of size at most $t$. In this case, we prove that there is no $\NCTMp$ of size at most $k$ for $\B(G)$. We first prove that $W\subseteq T(u_{m+1},1)$ for any $\NCTMp$ $T$ for $\B(G)$. Indeed, for all $j\in [m]$, $B_1(u_{m+1})=B_1(u_j)\cup \{w_j\}$, and so, to ensure that $T$ satisfies the non-clashing condition for the pair $B_1(u_{m+1})$ and $B_1(u_j)$, we must have that $w_j\in T(u_{m+1},1)$. Now, we prove that $|T(u_{m+1},1)\cap S|>t$ for any $\NCTMp$ $T$ for $\B(G)$ in this case, which completes the proof. Observe that, for all $i\in [n]$, $B_1(v_i)\subset B_1(u_{m+1})$ and $(B_1(u_{m+1})\setminus B_1(v_i))\subset S$. Hence, for each $i\in [n]$, to ensure that $T$ satisfies the non-clashing condition for the pair $B_1(u_{m+1})$ and $B_1(v_i)$, it is necessary that $s_j\in T(u_{m+1},1)$ for some $j\in [m]$ such that $s_j\notin B_1(v_i)$. However, $s_j\notin B_1(v_i)$ if and only if $i\in S_j$ in $\phi$. In other words, $T(u_{m+1},1)\cap S$ must correspond to a set cover in $\phi$, and thus, $|T(u_{m+1},1)|>k$. This concludes the proof for split graphs.

We now proceed with the proof for co-bipartite graphs. Let $\phi$ be an instance of \setcover\ with $X=\{ 1,\ldots,n\}$ and $\mathcal{S}=\{ S_1,\ldots,S_m\}$. As in the proof for split graphs, we may assume that each element of $X$ is contained in at most $m-2$ sets of $\mathcal{S}$. We may also assume that $m>n$ since we can simply duplicate sets in $\mathcal{S}$ to ensure this. From $\phi$, we construct the graph $G$ as in the proof for split graphs, except that we add the necessary edges so that the vertices in $W \cup S$ form a clique, and we add a vertex $v^*$ and make it adjacent to each vertex in $V \cup W \cup U$. The graph $G$ is clearly a co-bipartite graph with a universal vertex, that is constructed in polynomial time. See Figure~\ref{fig:cobip} for an illustration of $G$. We prove that $\phi$ admits a set cover of size at most $t$ if and only if there is an $\NCTMp$ of size at most $k=2m+t+1$ for $\B(G)$.

\begin{figure}[htb]
\centering
\includegraphics[scale=0.75]{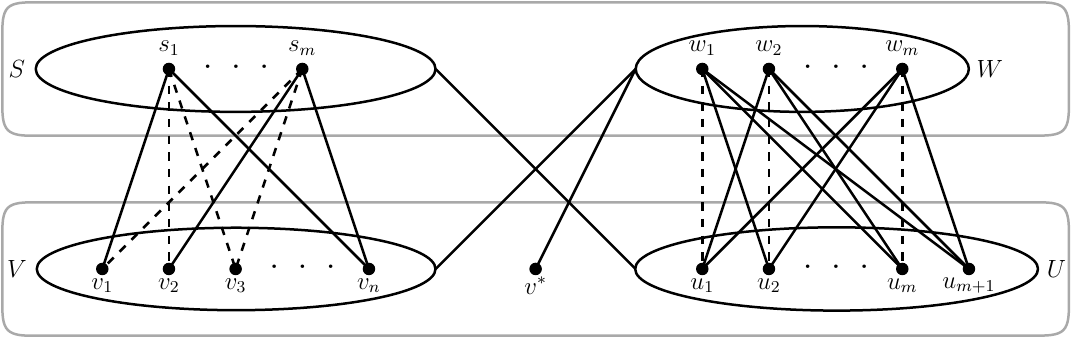}
\caption{\label{fig:cobip}
The co-bipartite graph $G$ constructed in the proof of Theorem~\ref{thm:split}. An edge between a vertex and an ellipse indicates that vertex is adjacent to each vertex in the ellipse.
}
\end{figure}

First, suppose that $\phi$ admits a set cover of size at most $t$, and let $\mathcal{S}'\subset \mathcal{S}$ be such a set cover. Let $S'\subset S$ be such that, for all $j\in [m]$, $S_j\in \mathcal{S}'$ in $\phi$ if and only if $s_j\in S'$. We define an $\NCTMp$ $T$ of size at most $k$ for $\B(G)$ as follows. As in the proof for split graphs, we only need to define $T$ for balls of $G$ of radius at most~$1$ since, for all $x\in V(G)$, $B_2(x)=B_1(u_{m+1})=V(G)$.

\begin{itemize}
\item For all $x\in V(G)$, set $T(x,0):=\{x\}$.
\item Set $T(v^*,1):=\{v^*,u_{m+1}\}$.
\item For all $x\in V$, set $T(x,1):=\{v^*\}\cup U \cup (B_1(x)\cap S)$.
\item For all $x\in S$, set $T(x,1):=\{x,u_{m+1}\} \cup (B_1(x)\cap V)$.
\item For all $x\in W$, set $T(x,1):=\{v^*\}\cup S \cup (B_1(x) \cap U)$.
\item Finally, for all $x\in U$, set $T(x,1):=\{v^*\}\cup S'\cup \{u_1,\ldots,u_m\} \cup (B_1(x)\cap W)$.
\end{itemize}

It is easy to verify that $T$ satisfies the inclusion condition for each ball in $\B(G)$, and that $T$ has size at most $k$ since $m>n$.
We now show that $T$ satisfies the non-clashing condition for all pairs of balls in $\B(G)$.
For all $x\in V(G)$, $|B_0(x)|=|T(x,0)|=1$ and $|T(x,1)|\geq 2$, and thus, $T$ satisfies the non-clashing condition for all pairs of balls in $\B(G)$ where at least one of the balls has radius~$0$. For all $x,y\in V$, $B_1(x)\setminus S=B_1(y)\setminus S$, $T(x,1)\cap S=B_1(x)\cap S$, and $T(y,1)\cap S=B_1(y)\cap S$. Hence, $T$ satisfies the non-clashing condition for all pairs of balls of radius~$1$ in $\B(G)$ centered at vertices in $V.$
For all $x,y\in U$, $B_1(x)\setminus W=B_1(y)\setminus W$, $T(x,1)\cap W=B_1(x)\cap W$, and $T(y,1)\cap W=B_1(y)\cap W$. Hence, $T$ satisfies the non-clashing condition for all pairs of balls of radius~$1$ in $\B(G)$ centered at vertices in $U.$
For all $x\in V$, by the construction, we have that $B_1(x)\cap S'\neq S'$ since $S'$ corresponds to the set cover $\mathcal{S}'$, and, for all $y \in U$, $S'\subset T(y,1)$. Hence, $T$ satisfies the non-clashing condition for all pairs of balls of radius~$1$ in $\B(G)$ centered at vertices in $V \cup U.$
For all $x,y\in W$, $B_1(x)\setminus U=B_1(y)\setminus U$, $T(x,1)\cap U=B_1(x)\cap U$, and $T(y,1)\cap U=B_1(y)\cap U$. Hence, $T$ satisfies the non-clashing condition for all pairs of balls of radius~$1$ in $\B(G)$ centered at vertices in $W.$
For all $x\in V \cup U$ and $y\in W$, we have that $\{u_1,\ldots,u_m\}\subset T(x,1)$ and $B_1(y)\cap \{u_1,\ldots,u_m\} \neq \{u_1,\ldots,u_m\}$. Hence, $T$ satisfies the non-clashing condition for all pairs of balls of radius~$1$ in $\B(G)$ centered at vertices in $V \cup U \cup W.$
For all $x\in V \cup U \cup W$, $|T(x,1)\cap S|\geq 1$ (recall that each element of $X$ is contained in at most $m-2$ sets of $\mathcal{S}$), and $|B_1(v^*)\cap S|=0$. Hence, $T$ satisfies the non-clashing condition for all pairs of balls of radius~$1$ in $\B(G)$ centered at vertices in $V \cup U \cup W \cup \{v^*\}.$
For all $x,y\in S$, $B_1(x)\setminus V=B_1(y)\setminus V$, $T(x,1)\cap V=B_1(x)\cap V$, and $T(y,1)\cap V=B_1(y)\cap V$. Hence, $T$ satisfies the non-clashing condition for all pairs of balls of radius~$1$ in $\B(G)$ centered at vertices in $S.$
For all $x\in V \cup U \cup W \cup \{v^*\}$ and $y\in S$, we have that $v^*\in T(x,1)$ and $v^*\notin B_1(y)$. Hence,
$T$ satisfies the non-clashing condition for all pairs of balls of radius~$1$ in $\B(G)$ centered at vertices in $V(G).$
Thus, $T$ is an $\NCTMp$ of size $\le k$ for $\B(G)$.

Now, suppose that $\phi$ does not admit a set cover of size at most $t$. In this case, we prove that there is no $\NCTMp$ of size at most $k$ for $\B(G)$. For all $x\in V \cup U$, it holds that $B_1(x)\setminus \{v^*\}$ is the same as in the graph $G$ constructed in the proof for split graphs. Thus, since, for all $x\in V \cup U$, it holds that $v^*\in B_1(x)$, we get that $|T(u_{m+1},1)\cap (W \cup S)| > m+t$ for any $\NCTMp$ $T$ for $\B(G)$ in this case. Now, we prove that $|T(u_{m+1},1)\cap \{u_1,\ldots,u_m\}|=m$ for any $\NCTMp$ $T$ for $\B(G)$. Indeed, for all $j\in [m]$, $B_1(u_{m+1})=B_1(w_j)\cup \{u_j\}$, and so, to ensure that $T$ satisfies the non-clashing condition for the pair $B_1(u_{m+1})$ and $B_1(w_j)$, we must have that $u_j\in T(u_{m+1},1)$. Lastly, we prove that $|T(u_{m+1},1)\cap (V \cup \{v^*\})|\geq 1$ for any $\NCTMp$ $T$ for $\B(G)$. Indeed, $B_1(u_{m+1})=V(G)$ and, for all $x\in S$, $B_1(u_{m+1})\setminus (V\cup \{v^*\})=B_1(x)\setminus (V\cup \{v^*\})$, and so, to ensure that $T$ satisfies the non-clashing condition for the pair $B_1(u_{m+1})$ and $B_1(x)$, we must have that $|T(u_{m+1},1)\cap (V\cup \{v^*\})|\geq 1$. Thus, $|T(u_{m+1},1)|>k$. This concludes the proof for co-bipartite graphs.

We now proceed with the proof for bipartite graphs. Let $\phi$ be an instance of \setcover\ with $X=\{ 1,\ldots,n\}$ and $\mathcal{S}=\{ S_1,\ldots,S_m\}$. As in the proof for co-bipartite graphs, we may assume that $m>n$, and that $\phi$ is an instance in which each element of $X$ is contained in at most $m-2$ sets of $\mathcal{S}$. From this instance $\phi$, we construct the graph $G$ as follows. For all $i\in [n]$ and $j\in [m]$, add a vertex $v_i$ and a vertex $s_j$, and if $i\notin S_j$ in $\phi$, then add the edge $v_is_j$. Add the sets of vertices $U=\{u_1,\ldots,u_{m+1}\}$ and $W=\{w_1,\ldots,w_m\}$, and, for all $j,\ell \in [m]$ such that $j\neq \ell$, add the edge $u_jw_{\ell}$. For all $j\in [m]$, add the edge $u_{m+1}w_j$. Add edges so that every vertex in $U$ is adjacent to every vertex in $V=\{v_1,\ldots,v_n\}$. Lastly, add a vertex $z$, and add edges so that $z$ is adjacent to every vertex in $U \cup S$, where $S=\{s_1,\ldots,s_m\}$. This completes the construction of $G$, which is clearly achieved in polynomial time. See Figure~\ref{fig:bipartite} for an illustration of $G$.  Note that $G$ has diameter~$3$ and is bipartite, as witnessed by the bipartition $(W \cup V \cup \{z\})\cup (U \cup S)$ of its vertices. We prove that $\phi$ admits a set cover of size at most $t$ if and only if there is an $\NCTMp$ of size at most $k=m+t$ for $\B(G)$.

\begin{figure}[htb]
\centering
\includegraphics[scale=0.75]{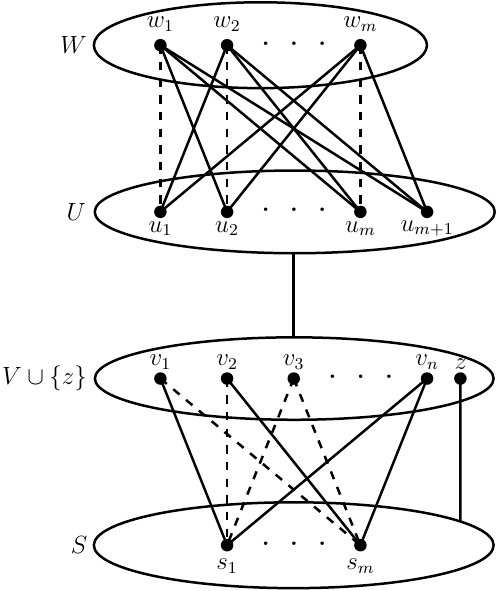}
\caption{\label{fig:bipartite}
The bipartite graph $G$ constructed in the proof of Theorem~\ref{thm:split}. An edge between a vertex and an ellipse indicates that vertex is adjacent to each vertex in the ellipse. See the caption of Figure~\ref{fig:split} for more details.
}
\end{figure}

First, suppose that $\phi$ admits a set cover of size at most $t$, and let $\mathcal{S}'\subset \mathcal{S}$ be such a set cover. Let $S'\subset S$ be such that, for all $j\in [m]$, $S_j\in \mathcal{S}'$ in $\phi$ if and only if $s_j\in S'$. We define an $\NCTMp$ $T$ of size at most $k$ for $\B(G)$ as follows. Also, note that we only need to define $T$ for balls of $G$ of radius at most~$2$ since, for all $x\in V(G)$, $B_3(x)=B_2(u_{m+1})=V(G)$. Furthermore, we only need to define $T$ for balls of radius at most~$1$ centered at $z$ since $B_2(z)=B_2(u_{m+1})$ and we will define $T$ for $B_2(u_{m+1})$.

\begin{itemize}
\item For all $x\in V(G)$, set $T(x,0):=\{x\}$.
\item For all $x\in V$, set $T(x,1):=\{x\}\cup (B_1(x)\cap S)$ and $T(x,2):=\{w_1\}\cup (B_1(x)\cap S)$, and note that $|T(x,2)|=|T(x,1)|\geq 3$ since, in $\phi$, every element of $X$ is contained in at most $m-2$ sets of $\mathcal{S}$.
\item Similarly, set $T(z,1):=\{z\}\cup S$.
\item For all $x\in S$, set $T(x,1):=\{x,z\} \cup (B_1(x)\cap V)$, $T(x,2):=\{u_{m+1},z\} \cup (B_1(x)\cap V)$.
\item For all $x\in W$, set $T(x,1):=\{x\}\cup (B_1(x)\cap \{u_1,\ldots,u_m\})$, and $T(x,2):=\{x,z\} \cup (B_1(x)\cap \{u_1,\ldots,u_m\})$.
\item Finally, for all $x\in U$, set $T(x,1):=\{x\}\cup (B_1(x)\cap W)$ and $T(x,2):=S'\cup (B_1(x)\cap W)$.
\end{itemize}

It is easy to verify that $T$ has size at most $k$ and satisfies the inclusion condition for each ball in $\B(G)$.

We now show that $T$ satisfies the non-clashing condition for all pairs of balls in $\B(G)$. For all $x\in V(G)$, $|B_0(x)|=|T(x,0)|=1$, $|T(x,1)|\geq 2$, and $|T(x,2)|\geq 2$, and thus, $T$ satisfies the non-clashing condition for all pairs of balls in $\B(G)$ where at least one of the balls has radius~$0$. For all $x\in V(G)$, we have that $x\in T(x,1)$. Furthermore, $W\cup S$, $W \cup V \cup \{z\}$, and $U \cup S$ are independent sets. Hence, $T$ satisfies the non-clashing condition for all pairs of balls of radius~$1$ in $\B(G)$ centered at vertices in
\begin{equation}\label{eqn:bip:1}
W \cup S;
\end{equation}
\begin{equation}\label{eqn:bip:2}
W \cup V \cup \{z\};
\end{equation}
\begin{equation}\label{eqn:bip:3}
U \cup S.
\end{equation}
For all $x\in W$ and $y \in U$, $|T(x,1)\cap U|=m-1$ and $|B_1(y)\cap U|=1$. Similarly, for all $x'\in V \cup \{z\}$ and $y'\in U \cup S$, $|T(x',1)\cap S|\geq 2$ and $|B_1(y')\cap S|\leq 1$. Hence, in combination with~(\ref{eqn:bip:1}),~(\ref{eqn:bip:2}),~and~(\ref{eqn:bip:3}), $T$ satisfies the non-clashing condition for all pairs of balls of radius~$1$ in $\B(G)$ centered at vertices in
\begin{equation}\label{eqn:bip:4}
V(G).
\end{equation}
For all $x\in W$ and $y\in U$, $|T(x,2)\cap W|=1$, $|T(y,2)\cap W|\geq m-1$, and $|B_1(z)\cap W|=0$. Also, for all $x'\in V$, $|B_2(x')\cap S|<m$ and $|T(z,1)\cap S|=m$. Lastly, $|B_1(z)\cap V|=0$ and, for all $y'\in S$, either $|T(y',2)\cap V|=0$, in which case $B_1(z)=B_2(y')$, or $|T(y',2)\cap V|\geq 1$. Recall that $B_2(u_{m+1})=B_2(z)$, and hence, $T$ satisfies the non-clashing condition for all pairs of balls in $\B(G)$, where one of the balls is $B_1(z)$ and the other has radius~$2$ and is centered at a vertex in
\begin{equation}\label{eqn:bip:5}
V(G).
\end{equation}
For all $x\in V(G)$, there exists a vertex $y\in T(x,2)$ such that $d(x,y)=2$. Hence, in combination with the construction of $G$, $T$ satisfies the non-clashing condition for all pairs of balls in $\B(G)$, where one of the balls has radius~$1$ and the other has radius~$2$, that are centered at vertices in
\begin{equation}\label{eqn:bip:6}
W;
\end{equation}
\begin{equation}\label{eqn:bip:7}
U;
\end{equation}
\begin{equation}\label{eqn:bip:8}
V;
\end{equation}
\begin{equation}\label{eqn:bip:9}
S.
\end{equation}
For all $x\in S$ and $y\in W$, we have that $z\in T(x,2)$, $z\notin B_1(y)$, $|T(y,2)\cap U|=m-1$, and $|B_1(x)\cap U|=0$. Hence, in combination with~(\ref{eqn:bip:6}) and (\ref{eqn:bip:9}), $T$ satisfies the non-clashing condition for all pairs of balls in $\B(G)$, where one of the balls has radius~$1$ and the other has radius~$2$, that are centered at vertices in
\begin{equation}\label{eqn:bip:10}
W \cup S.
\end{equation}
For all $x\in V$, $y\in W$, and $q\in \{1,2\}$, $|T(x,q)\cap S|\geq 2$ and $|B_q(y)\cap S|=0$. Hence, in combination with~(\ref{eqn:bip:6}) and (\ref{eqn:bip:8}), $T$ satisfies the non-clashing condition for all pairs of balls in $\B(G)$, where one of the balls has radius~$1$ and the other has radius~$2$, that are centered at vertices in
\begin{equation}\label{eqn:bip:11}
W \cup V.
\end{equation}
For all $x\in W$ and $y\in U$, $|T(x,2)\cap U|=m-1$, $|B_1(y)\cap U|=1$, $|T(y,2)\cap S|\geq 1$, and $|B_1(x)\cap S|=0$. Hence, in combination with~(\ref{eqn:bip:6}) and (\ref{eqn:bip:7}), $T$ satisfies the non-clashing condition for all pairs of balls in $\B(G)$, where one of the balls has radius~$1$ and the other has radius~$2$, that are centered at vertices in
\begin{equation}\label{eqn:bip:12}
W \cup U.
\end{equation}
For all $x\in S$, $y \in U$, and $q\in \{1,2\}$, $|T(y,q)\cap W|\geq m-1$ and $|B_q(x)\cap W|=0$. Hence, in combination with~(\ref{eqn:bip:7}) and (\ref{eqn:bip:9}), $T$ satisfies the non-clashing condition for all pairs of balls in $\B(G)$, where one of the balls has radius~$1$ and the other has radius~$2$, that are centered at vertices in
\begin{equation}\label{eqn:bip:13}
U \cup S.
\end{equation}
For all $x\in V$ and $y \in U$, $|T(y,2)\cap W|\geq m-1$, $|B_1(x)\cap W|=0$, $|T(x,2)\cap S|\geq 2$, and $|B_1(y)\cap S|=0$. Hence, in combination with~(\ref{eqn:bip:7}) and (\ref{eqn:bip:8}), $T$ satisfies the non-clashing condition for all pairs of balls in $\B(G)$, where one of the balls has radius~$1$ and the other has radius~$2$, that are centered at vertices in
\begin{equation}\label{eqn:bip:14}
U \cup V.
\end{equation}
For all $x\in V$ and $y\in S$, we have that $w_1\in T(x,2)$, $w_1\notin B_1(y)$, $z\in T(y,2)$, and $z\notin B_1(x)$. Hence, in combination with~(\ref{eqn:bip:8}) and (\ref{eqn:bip:9}), $T$ satisfies the non-clashing condition for all pairs of balls in $\B(G)$, where one of the balls has radius~$1$ and the other has radius~$2$, that are centered at vertices in
\begin{equation}\label{eqn:bip:15}
V \cup S.
\end{equation}
Combining~(\ref{eqn:bip:5}) and (\ref{eqn:bip:10})--(\ref{eqn:bip:15}), $T$ satisfies the non-clashing condition for all pairs of balls in $\B(G)$, where one of the balls has radius~$1$ and the other has radius~$2$, that are centered at vertices in
\begin{equation}\label{eqn:bip:16}
V(G).
\end{equation}
It remains to prove that $T$ satisfies the non-clashing condition for all pairs of balls of radius~$2$ in $\B(G)$. For all $j\in [m]$, $T(w_j,2)\cap \{u_1,\ldots,u_m\}=B_2(w_j)\cap \{u_1,\ldots,u_m\}=\{u_1\ldots,u_m\}\setminus \{u_j\}$. Hence, $T$ satisfies the non-clashing condition for all pairs of balls of radius~$2$ in $\B(G)$, that are centered at vertices in
\begin{equation}\label{eqn:bip:17}
W.
\end{equation}
For all $j\in [m]$, $T(u_j,2)\cap W=B_2(u_j)\cap W = W\setminus \{w_j\}$. Also, $W \subset T(u_{m+1},2)$. Hence, $T$ satisfies the non-clashing condition for all pairs of balls of radius~$2$ in $\B(G)$, that are centered at vertices in
\begin{equation}\label{eqn:bip:18}
U.
\end{equation}
For any $x,y\in V$, $B_2(x)\setminus S=B_2(y)\setminus S$, $T(x,2)\cap S=B_2(x)\cap S$, and $T(y,2)\cap S=B_2(y)\cap S$.
Hence, $T$ satisfies the non-clashing condition for all pairs of balls of radius~$2$ in $\B(G)$, that are centered at vertices in
\begin{equation}\label{eqn:bip:19}
V.
\end{equation}
For any $x,y\in S$, $B_2(x)\setminus V=B_2(y)\setminus V$, $T(x,2)\cap V=B_2(x)\cap V$, and $T(y,2)\cap V=B_2(y)\cap V$.
Hence, $T$ satisfies the non-clashing condition for all pairs of balls of radius~$2$ in $\B(G)$, that are centered at vertices in
\begin{equation}\label{eqn:bip:20}
S.
\end{equation}
For all $x\in W$ and $y\in S$, we have that $x\in T(x,2)$ and $|B_2(y)\cap W|=0$. Hence, in combination with~(\ref{eqn:bip:17}) and (\ref{eqn:bip:20}), $T$ satisfies the non-clashing condition for all pairs of balls of radius~$2$ in $\B(G)$, that are centered at vertices in
\begin{equation}\label{eqn:bip:21}
W \cup S.
\end{equation}
For all $x\in V$ and $y\in W$, $|T(x,2)\cap S|\geq 2$ and $|B_2(y)\cap S|=0$. Hence, in combination with~(\ref{eqn:bip:17}) and (\ref{eqn:bip:19}), $T$ satisfies the non-clashing condition for all pairs of balls of radius~$2$ in $\B(G)$, that are centered at vertices in
\begin{equation}\label{eqn:bip:22}
W \cup V.
\end{equation}
For all $x\in W$ and $y\in U$, $|T(y,2)\cap S|\geq 1$ and $|B_2(x)\cap S|=0$. Hence, in combination with~(\ref{eqn:bip:17}) and (\ref{eqn:bip:18}), $T$ satisfies the non-clashing condition for all pairs of balls of radius~$2$ in $\B(G)$, that are centered at vertices in
\begin{equation}\label{eqn:bip:23}
W \cup U.
\end{equation}
For all $x\in S$ and $y \in U$, $|T(y,2)\cap W|\geq m-1$ and $|B_2(x)\cap W|=0$. Hence, in combination with~(\ref{eqn:bip:18}) and (\ref{eqn:bip:20}), $T$ satisfies the non-clashing condition for all pairs of balls of radius~$2$ in $\B(G)$, that are centered at vertices in
\begin{equation}\label{eqn:bip:24}
U \cup S.
\end{equation}
For all $x \in V$ and $y\in S$, we have that $w_1\in T(x,2)$ and $w_1\notin B_2(y)$. Hence, in combination with~(\ref{eqn:bip:19}) and (\ref{eqn:bip:20}), $T$ satisfies the non-clashing condition for all pairs of balls of radius~$2$ in $\B(G)$, that are centered at vertices in
\begin{equation}\label{eqn:bip:25}
V \cup S.
\end{equation}
For all $x\in V$, by the construction, $B_2(x)\cap S'\neq S'$ since $S'$ corresponds to the set cover $\mathcal{S}'$, and, for all $y \in U$, $S'\subset T(y,2)$. Hence, in combination with~(\ref{eqn:bip:18}) and (\ref{eqn:bip:19}), $T$ satisfies the non-clashing condition for all pairs of balls of radius~$2$ in $\B(G)$, that are centered at vertices in
\begin{equation}\label{eqn:bip:26}
U \cup V.
\end{equation}
Combining~(\ref{eqn:bip:21})--(\ref{eqn:bip:26}), we get that $T$ satisfies the non-clashing condition for all pairs of balls of radius~$2$ in $\B(G)$, that are centered at vertices in
\begin{equation}\label{eqn:bip:27}
V(G).
\end{equation}
Combining~(\ref{eqn:bip:4}), (\ref{eqn:bip:16}), and (\ref{eqn:bip:27}), we get that $T$ is an $\NCTMp$ of size at most $k$ for $\B(G)$.

Now, suppose that $\phi$ does not admit a set cover of size at most $t$. In this case, we prove that there is no $\NCTMp$ of size at most $k$ for $\B(G)$. We first prove that $|T(u_{m+1},2)\cap W|=m$ for any $\NCTMp$ $T$ for $\B(G)$. Indeed, for all $1\leq j\leq m$, $B_2(u_{m+1})=B_2(u_j)\cup \{w_j\}$, and so, to ensure that $T$ satisfies the non-clashing condition for the pair $B_2(u_{m+1})$ and $B_2(u_j)$, we must have that $w_j\in T(u_{m+1},2)$. Now, we prove that $|T(u_{m+1},2)\cap S|>t$ for any $\NCTMp$ $T$ for $\B(G)$ in this case, which completes the proof. Observe that, for all $i\in [n]$, $B_2(v_i)\subset B_2(u_{m+1})$ and $(B_2(u_{m+1})\setminus B_2(v_i))\subset S$. Hence, for each $i\in [n]$, to ensure that $T$ satisfies the non-clashing condition for the pair $B_2(u_{m+1})$ and $B_2(v_i)$, it is necessary that $s_j\in T(u_{m+1},1)$ for some $j\in [m]$ such that $s_j\notin B_1(v_i)$. However, $s_j\notin B_1(v_i)$ if and only if $i\in S_j$ in $\phi$. In other words, $T(u_{m+1},2)\cap S$ must correspond to a set cover in $\phi$, and thus, $|T(u_{m+1},2)|>k$. This concludes the proof for bipartite graphs.
\end{proof}

\setcounter{equation}{0}

\propexact*

\begin{proof}
For any $x\in V(G)$ and $r\in \mathbb{N}$, there are at most $2^n$ possible choices for $T(x,r)$, and there are at most $n \cdot \min\{(\diam + 1),n\}$ unique balls in $G$. Thus, for each possible (positive) $\NCTM$, it can be checked in polynomial time whether it satisfies the non-clashing condition for all pairs of balls in $\B(G)$, and the inclusion condition (for \srmfull) for each ball in $\B(G)$. Hence, there is a brute-force algorithm running in time $2^{\mathcal{O}(n^2 \cdot \min\{\diam+1,n\})}=2^{\mathcal{O}(n^2 \cdot \diam(G))}$ (since $G$ is connected).
\end{proof}

\mainforward*

\begin{proof}
Suppose that $\pi: X^{\alpha} \cup X^{\beta} \cup X^{\gamma} \rightarrow \{\text{True},\text{False}\}$ is a satisfying assignment for $\phi$.
Let us define the set $\pi'$ of vertices in $A^{\alpha} \cup A^{\beta} \cup A^{\gamma}$ corresponding to $\pi$.
Initially, set $\pi':=\varnothing$.
Now, for each $\delta\in \{\alpha,\beta,\gamma\}$ and $x^{\delta}_i$ in $\phi$, if $\pi(x^{\delta}_i)=\text{True}$ ($\pi(x^{\delta}_i)=\text{False}$, respectively), then add $t^{\delta}_{2i}$ ($f^{\delta}_{2i-1}$, respectively) to $\pi'$.
Thus, $|\pi'|=3N$ and $\pi'$ corresponds to a satisfying assignment for $\phi$ in the sense that, from $\pi'$, we can extract the satisfying assignment $\pi$ for $\phi$.
Using $\pi'$, we define an $\NCTMp$ $T$ of size $k$ for $\B(G)$ as follows. Also, note that we only need to define $T$ for balls of $G$ of radius at most~$2$ since, for all $x\in V(G)$, $B_3(x)=B_2(u_{3M+1})=V(G)$. Furthermore, we do not need to define $T$ for $B_2(u'_{3M+1})$ nor $B_2(z)$ since $B_2(z)=B_2(u'_{3M+1})=B_2(u_{3M+1})$, and we will define $T$ for $B_2(u_{3M+1})$.

\begin{itemize}
\item For all $x\in V(G)$, set $T(x,0):=\{x\}$.
\item For each $\delta\in \{\alpha,\beta,\gamma\}$ and $x\in A^{\delta}$, set $T(x,1):=B_1(x)$ and $T(x,2):=\{u_1,t^{\delta'}_2,t^{\delta''}_2\} \cup B_1(x)$, where $\delta',\delta''\in \{\alpha,\beta,\gamma\}$ such that $\delta\notin \{\delta',\delta''\}$ and $\delta'\neq \delta''$. Note that $T(x,1)\subset T(x,2)$ and $|T(x,2)|=\mathcal{O}(\log M)$.
\item For each $\delta\in \{\alpha,\beta,\gamma\}$ and $x\in V^{\delta} \cup V^{\delta,*}$, set $T(x,1):=B_1(x)\setminus U$ and $T(x,2):=\{w_1,z\} \cup (B_1(x)\setminus U)$. Note that $T(x,1)\subset T(x,2)$ and $|T(x,2)|<2N+3M+3$.
\item For each $x\in V^W$, set $T(x,1):=\{u_{3M+1},u'_{3M+1}\} \cup V^W \cup (B_1(x)\cap U)$ and $T(x,2):=\{u_{3M+1},u'_{3M+1},z\} \cup V^W \cup U$. Note that $T(x,1)\subset T(x,2)$ and $|T(x,2)|=3M+\mathcal{O}(\log M)$.
\item For each $\delta\in \{\alpha,\beta,\gamma\}$ and $x\in C \cup C^{\delta}$, set $T(x,1):=B_1(x)$ and $T(x,2):=\{x,w_1\} \cup (B_2(x)\cap (A^{\alpha} \cup A^{\beta} \cup A^{\gamma}))$. Note that $|T(x,1)|=\mathcal{O}(\log M)$ and $|T(x,2)|\leq 6N+2 < 3N+3M=k$.
\item For each $x\in U \cup \{u_{3M+1},u'_{3M+1}\}$, set $T(x,1):=B_1(x)\setminus (C \cup C^{\alpha} \cup C^{\beta} \cup C^{\gamma})$ and $T(x,2):=(B_2(x)\cap W) \cup \pi'$. Note that $|T(x,1)|\leq 3M+\mathcal{O}(\log M)$ and $|T(x,2)|\leq  3N+3M = k$.
\item For each $x\in W$, set $T(x,1):=B_1(x)$ and $T(x,2):=\{x,z,u_{3M+1},u'_{3M+1}\} \cup (B_2(x)\cap U)$. Note that $|T(x,1)|=\mathcal{O}(\log M)$ and $|T(x,2)|\leq 3M+4$.
\item Finally, set $T(z,1):=\{z,u_1\}$.
\end{itemize}

Hence, $T$ has size at most $k$, and it is easy to verify that $T$ satisfies the inclusion condition for each ball in $\B(G)$.

We now show that $T$ satisfies the non-clashing condition for all pairs of balls in $\B(G)$. For all $x\in V(G)$, $|B_0(x)|=|T(x,0)|=1$, $|T(x,1)|\geq 2$, and $|T(x,2)|\geq 2$, and thus, $T$ satisfies the non-clashing condition for all pairs of balls in $\B(G)$ where at least one of the balls has radius~$0$. For all $x\in V(G)$, $x\in T(x,1)$. Furthermore, $W \cup U \cup C \cup C^{\alpha} \cup C^{\beta} \cup C^{\gamma} \cup A^{\alpha} \cup A^{\beta} \cup A^{\gamma}$, $\{z\} \cup W \cup V^{\alpha} \cup V^{\beta} \cup V^{\gamma} \cup V^{\alpha,*} \cup V^{\beta,*} \cup V^{\gamma,*}$, $\{u_{3M+1},u'_{3M+1}\} \cup U \cup A^{\alpha} \cup A^{\beta} \cup A^{\gamma}$, $C \cup V^{\alpha,*} \cup V^{\beta,*} \cup V^{\gamma,*}$, $C^{\alpha} \cup C^{\beta} \cup C^{\gamma} \cup V^{\alpha} \cup V^{\beta} \cup V^{\gamma}$, and $\{z\} \cup C \cup C^{\alpha} \cup C^{\beta} \cup C^{\gamma}$ are independent sets. Hence, $T$ satisfies the non-clashing condition for all pairs of balls of radius~$1$ in $\B(G)$ centered at vertices in
\begin{equation}\label{eqn:vc:1}
W \cup U \cup C \cup C^{\alpha} \cup C^{\beta} \cup C^{\gamma} \cup A^{\alpha} \cup A^{\beta} \cup A^{\gamma};
\end{equation}
\begin{equation}\label{eqn:vc:2}
\{z\} \cup W \cup V^{\alpha} \cup V^{\beta} \cup V^{\gamma} \cup V^{\alpha,*} \cup V^{\beta,*} \cup V^{\gamma,*};
\end{equation}
\begin{equation}\label{eqn:vc:3}
\{u_{3M+1},u'_{3M+1}\} \cup U \cup A^{\alpha} \cup A^{\beta} \cup A^{\gamma};
\end{equation}
\begin{equation}\label{eqn:vc:4}
C \cup V^{\alpha,*} \cup V^{\beta,*} \cup V^{\gamma,*};
\end{equation}
\begin{equation}\label{eqn:vc:5}
C^{\alpha} \cup C^{\beta} \cup C^{\gamma} \cup V^{\alpha} \cup V^{\beta} \cup V^{\gamma};
\end{equation}
\begin{equation}\label{eqn:vc:6}
\{z\} \cup C \cup C^{\alpha} \cup C^{\beta} \cup C^{\gamma}.
\end{equation}
For all $x\in W$ and $y\in V^W$, $|B_1(x)\cap V^W|=p$, $|T(y,1)\cap V^W|=2p$, and $|T(u_{3M+1},1)\cap W|=|T(u'_{3M+1},1)\cap W|=3M$. Hence, in combination with (\ref{eqn:vc:1}) and (\ref{eqn:vc:2}), $T$ satisfies the non-clashing condition for all pairs of balls of radius~$1$ in $\B(G)$, where one of the balls is centered at a vertex in
\begin{equation}\label{eqn:vc:7}
W.
\end{equation}
For all $x\in V^W$ and $x'\in U \cup \{u_{3M+1},u'_{3M+1}\}$, we have that $z\in T(x',1)$ and $z\notin B_1(x)$. Further, for all $x\in V^W$ and $x''\in V(G)\setminus (V^W \cup W \cup U \cup \{u_{3M+1},u'_{3M+1}\})$, we have that $x''\in T(x'',1)$ and $x''\notin B_1(x)$. Lastly, for all $x,y\in V^W$, $B_1(x)\setminus (W \cup U) = B_1(y)\setminus (W \cup U)$, $(B_1(x)\cap U) \subset T(x,1)$, $(B_1(y)\cap U) \subset T(y,1)$, and $B_1(x)\cap U = B_1(y)\cap U$ if and only if $B_1(x)\cap W = B_1(y)\cap W$ by the construction.  Hence, in combination with (\ref{eqn:vc:7}), $T$ satisfies the non-clashing condition for all pairs of balls of radius~$1$ in $\B(G)$, where one of the balls is centered at a vertex in
\begin{equation}\label{eqn:vc:8}
V^W.
\end{equation}
For all $x\in U \cup \{u_{3M+1},u'_{3M+1}\}$, $|T(x,1)\cap V^W|\geq p$. Further, for all $y\in V(G)\setminus (V^W \cup W \cup U \cup \{u_{3M+1},u'_{3M+1}\})$, $|B_1(y)\cap V^W|=0$. Hence, in combination with (\ref{eqn:vc:3}), (\ref{eqn:vc:7}), and (\ref{eqn:vc:8}), this implies that $T$ satisfies the non-clashing condition for all pairs of balls of radius~$1$ in $\B(G)$, where one of the balls is centered at a vertex in
\begin{equation}\label{eqn:vc:9}
U \cup \{u_{3M+1},u'_{3M+1}\}.
\end{equation}
For all $x\in C$, $y\in C^{\alpha} \cup C^{\beta} \cup C^{\gamma}$, and $\delta\in \{\alpha,\beta,\gamma\}$, $|T(x,1)\cap V^{\delta}|\geq p$ and $|T(y,1)\cap V^{\delta,*}|\geq p$. Hence, in combination with (\ref{eqn:vc:1}), (\ref{eqn:vc:4}), (\ref{eqn:vc:5}), (\ref{eqn:vc:6}), (\ref{eqn:vc:8}), and (\ref{eqn:vc:9}), $T$ satisfies the non-clashing condition for all pairs of balls of radius~$1$ in $\B(G)$, where one of the balls is centered at a vertex in
\begin{equation}\label{eqn:vc:10}
C \cup C^{\alpha} \cup C^{\beta} \cup C^{\gamma}.
\end{equation}
For all $x\in A^{\alpha} \cup A^{\beta} \cup A^{\gamma}$ and $\delta\in \{\alpha,\beta,\gamma\}$, $|T(x,1)\cap V^{\delta}|\geq p$ and $z\in T(x,1)$. Hence, in combination with (\ref{eqn:vc:1}), (\ref{eqn:vc:2}), and (\ref{eqn:vc:7})--(\ref{eqn:vc:10}), $T$ satisfies the non-clashing condition for all pairs of balls of radius~$1$ in $\B(G)$ centered at vertices in
\begin{equation}\label{eqn:vc:11}
V(G).
\end{equation}
For all $x\in V(G)\setminus (U \cup A^{\alpha} \cup A^{\beta} \cup A^{\gamma} \cup \{u_{3M+1},u'_{3M+1}\})$, we have that $x\in T(x,2)$ and $x\notin B_1(z)$. For all $\delta\in \{\alpha,\beta,\gamma\}$ and $x'\in A^{\delta}$, $|T(x',2)\cap V^{\delta}|=p$ and $|B_1(z)\cap V^{\delta}|=0$. Lastly, for all $x''\in U \cup \{u_{3M+1}\}$, $|T(x'',2)\cap W|\geq 3M-1$ and $|B_1(z)\cap W|=0$.
Recall that $B_2(u_{3M+1})=B_2(u'_{3M+1})=B_2(z)$, and hence, $T$ satisfies the non-clashing condition for all pairs of balls in $\B(G)$, where one of the balls is $B_1(z)$ and the other has radius~$2$ and is centered at a vertex in
\begin{equation}\label{eqn:vc:12}
V(G).
\end{equation}
For all $x\in V(G)$, there exists a vertex $y\in T(x,2)$ such that $d(x,y)=2$. Hence, in combination with the construction of $G$, $T$ satisfies the non-clashing condition for all pairs of balls in $\B(G)$, where one of the balls has radius~$1$ and the other has radius~$2$, that are centered at vertices in
\begin{equation}\label{eqn:vc:13}
W;
\end{equation}
\begin{equation}\label{eqn:vc:14}
V^W;
\end{equation}
\begin{equation}\label{eqn:vc:15}
U \cup \{u_{3M+1},u'_{3M+1}\};
\end{equation}
\begin{equation}\label{eqn:vc:16}
C \cup C^{\alpha} \cup C^{\beta} \cup C^{\gamma};
\end{equation}
\begin{equation}\label{eqn:vc:17}
V^{\alpha} \cup V^{\beta} \cup V^{\gamma} \cup V^{\alpha,*} \cup V^{\beta,*} \cup V^{\gamma,*};
\end{equation}
\begin{equation}\label{eqn:vc:18}
A^{\alpha} \cup A^{\beta} \cup A^{\gamma}.
\end{equation}
For all $x\in W$ and $y\in V^W$, we have that $z\in T(x,2)$, $z\notin B_1(y)$, $z\in T(y,2)$, and $z\notin B_1(x)$. Hence, in combination with (\ref{eqn:vc:13}) and (\ref{eqn:vc:14}), $T$ satisfies the non-clashing condition for all pairs of balls in $\B(G)$, where one of the balls has radius~$1$ and the other has radius~$2$, that are centered at vertices in
\begin{equation}\label{eqn:vc:19}
W\cup V^W.
\end{equation}
For all $x\in W$ and $y\in U \cup \{u_{3M+1},u'_{3M+1}\}$, $|T(x,2)\cap \{u_{3M+1},u'_{3M+1}\}|=2$, $|B_1(y)\cap \{u_{3M+1},u'_{3M+1}\}|\leq 1$, $|T(y,2)\cap A^{\alpha}|=N$, and $|B_1(x)\cap A^{\alpha}|=0$. Hence, in combination with (\ref{eqn:vc:13}) and (\ref{eqn:vc:15}), $T$ satisfies the non-clashing condition for all pairs of balls in $\B(G)$, where one of the balls has radius~$1$ and the other has radius~$2$, that are centered at vertices in
\begin{equation}\label{eqn:vc:20}
W\cup U \cup \{u_{3M+1},u'_{3M+1}\}.
\end{equation}
For all $x\in W$ and $y\in V(G)\setminus (W \cup V^W\cup U \cup \{u_{3M+1},u'_{3M+1},z\})$, we have that $x\in T(x,2)$, $x\notin B_1(y)$, $y\in T(y,2)$, and $y\notin B_1(x)$. Hence, in combination with (\ref{eqn:vc:13}) and (\ref{eqn:vc:16})--(\ref{eqn:vc:18}), $T$ satisfies the non-clashing condition for all pairs of balls in $\B(G)$, where one of the balls has radius~$1$ and the other has radius~$2$, that are centered at vertices in
\begin{equation}\label{eqn:vc:21}
W\cup C \cup C^{\alpha} \cup C^{\beta} \cup C^{\gamma};
\end{equation}
\begin{equation}\label{eqn:vc:22}
W \cup V^{\alpha} \cup V^{\beta} \cup V^{\gamma} \cup V^{\alpha,*} \cup V^{\beta,*} \cup V^{\gamma,*};
\end{equation}
\begin{equation}\label{eqn:vc:23}
W \cup A^{\alpha} \cup A^{\beta} \cup A^{\gamma}.
\end{equation}
For all $x\in V^W$ and $y\in U \cup \{u_{3M+1},u'_{3M+1}\}$, $|T(x,2)\cap \{u_{3M+1},u'_{3M+1}\}|=2$, $|B_1(y)\cap \{u_{3M+1},u'_{3M+1}\}|\leq 1$, $|T(y,2)\cap A^{\alpha}|=N$, and $|B_1(x)\cap A^{\alpha}|=0$. Hence, in combination with (\ref{eqn:vc:14}) and (\ref{eqn:vc:15}), $T$ satisfies the non-clashing condition for all pairs of balls in $\B(G)$, where one of the balls has radius~$1$ and the other has radius~$2$, that are centered at vertices in
\begin{equation}\label{eqn:vc:24}
V^W\cup U \cup \{u_{3M+1},u'_{3M+1}\}.
\end{equation}
For all $x\in V^W$ and $y\in V(G)\setminus (W \cup V^W\cup U \cup \{u_{3M+1},u'_{3M+1},z\})$, we have that $x\in T(x,2)$, $x\notin B_1(y)$, $y\in T(y,2)$, and $y\notin B_1(x)$. Hence, in combination with (\ref{eqn:vc:14}) and (\ref{eqn:vc:16})--(\ref{eqn:vc:18}), $T$ satisfies the non-clashing condition for all pairs of balls in $\B(G)$, where one of the balls has radius~$1$ and the other has radius~$2$, that are centered at vertices in
\begin{equation}\label{eqn:vc:25}
V^W\cup C \cup C^{\alpha} \cup C^{\beta} \cup C^{\gamma};
\end{equation}
\begin{equation}\label{eqn:vc:26}
V^W \cup V^{\alpha} \cup V^{\beta} \cup V^{\gamma} \cup V^{\alpha,*} \cup V^{\beta,*} \cup V^{\gamma,*};
\end{equation}
\begin{equation}\label{eqn:vc:27}
V^W \cup A^{\alpha} \cup A^{\beta} \cup A^{\gamma}.
\end{equation}
For all $x\in U \cup \{u_{3M+1},u'_{3M+1}\}$ and $y\in V(G)\setminus (W \cup V^W\cup U \cup \{u_{3M+1},u'_{3M+1},z\})$, $|T(x,2)\cap W|\geq 3M-1$, $|B_1(y)\cap W|=0$, $|T(y,2)\cap (A^{\alpha} \cup A^{\beta} \cup A^{\gamma})|\geq 1$, and $|B_1(x)\cap (A^{\alpha} \cup A^{\beta} \cup A^{\gamma})|=0$. Hence, in combination with (\ref{eqn:vc:15})--(\ref{eqn:vc:18}), $T$ satisfies the non-clashing condition for all pairs of balls in $\B(G)$, where one of the balls has radius~$1$ and the other has radius~$2$, that are centered at vertices in
\begin{equation}\label{eqn:vc:28}
U \cup \{u_{3M+1},u'_{3M+1}\} \cup C \cup C^{\alpha} \cup C^{\beta} \cup C^{\gamma};
\end{equation}
\begin{equation}\label{eqn:vc:29}
U \cup \{u_{3M+1},u'_{3M+1}\} \cup V^{\alpha} \cup V^{\beta} \cup V^{\gamma} \cup V^{\alpha,*} \cup V^{\beta,*} \cup V^{\gamma,*};
\end{equation}
\begin{equation}\label{eqn:vc:30}
U \cup \{u_{3M+1},u'_{3M+1}\} \cup A^{\alpha} \cup A^{\beta} \cup A^{\gamma}.
\end{equation}
For all $x\in C \cup C^{\alpha} \cup C^{\beta} \cup C^{\gamma}$ and $y\in V^{\alpha} \cup V^{\beta} \cup V^{\gamma} \cup V^{\alpha,*} \cup V^{\beta,*} \cup V^{\gamma,*}$, we have that $w_1\in T(x,2)$, $w_1\notin B_1(y)$, $z\in T(y,2)$, and $z\notin B_1(x)$. Hence, in combination with (\ref{eqn:vc:16}) and (\ref{eqn:vc:17}), $T$ satisfies the non-clashing condition for all pairs of balls in $\B(G)$, where one of the balls has radius~$1$ and the other has radius~$2$, that are centered at vertices in
\begin{equation}\label{eqn:vc:31}
C \cup C^{\alpha} \cup C^{\beta} \cup C^{\gamma} \cup V^{\alpha} \cup V^{\beta} \cup V^{\gamma} \cup V^{\alpha,*} \cup V^{\beta,*} \cup V^{\gamma,*}.
\end{equation}
For all $x\in C \cup C^{\alpha} \cup C^{\beta} \cup C^{\gamma}$ and $y\in A^{\alpha} \cup A^{\beta} \cup A^{\gamma}$, we have that $x\in T(x,2)$, $x\notin B_1(y)$, $y\in T(y,2)$, and $y\notin B_1(x)$. Hence, in combination with (\ref{eqn:vc:16}) and (\ref{eqn:vc:18}), $T$ satisfies the non-clashing condition for all pairs of balls in $\B(G)$, where one of the balls has radius~$1$ and the other has radius~$2$, that are centered at vertices in
\begin{equation}\label{eqn:vc:32}
C \cup C^{\alpha} \cup C^{\beta} \cup C^{\gamma} \cup A^{\alpha} \cup A^{\beta} \cup A^{\gamma}.
\end{equation}
For all $x\in V^{\alpha} \cup V^{\beta} \cup V^{\gamma} \cup V^{\alpha,*} \cup V^{\beta,*} \cup V^{\gamma,*}$ and $y\in A^{\alpha} \cup A^{\beta} \cup A^{\gamma}$, we have that $w_1\in T(x,2)$, $w_1\notin B_1(y)$, $z\in T(y,2)$, and $z\notin B_1(x)$. Hence, in combination with (\ref{eqn:vc:17}) and (\ref{eqn:vc:18}), $T$ satisfies the non-clashing condition for all pairs of balls in $\B(G)$, where one of the balls has radius~$1$ and the other has radius~$2$, that are centered at vertices in
\begin{equation}\label{eqn:vc:33}
V^{\alpha} \cup V^{\beta} \cup V^{\gamma} \cup V^{\alpha,*} \cup V^{\beta,*} \cup V^{\gamma,*} \cup A^{\alpha} \cup A^{\beta} \cup A^{\gamma}.
\end{equation}
Combining (\ref{eqn:vc:12}) and (\ref{eqn:vc:19})--(\ref{eqn:vc:33}), $T$ satisfies the non-clashing condition for all pairs of balls in $\B(G)$, where one of the balls has radius~$1$ and the other has radius~$2$, that are centered at vertices in
\begin{equation}\label{eqn:vc:34}
V(G).
\end{equation}
It remains to prove that $T$ satisfies the non-clashing condition for all pairs of balls of radius~$2$ in $\B(G)$.
For all $\ell\in [3M]$, $T(w_{\ell},2)\cap U=B_2(w_{\ell})\cap U=U\setminus \{u_{\ell}\}$. Hence, $T$ satisfies the non-clashing condition for all pairs of balls of radius~$2$ in $\B(G)$ that are centered at vertices in
\begin{equation}\label{eqn:vc:35}
W.
\end{equation}
For any $x,y\in V^W$, $B_2(x) = B_2(y)$. Hence, $T$ satisfies the non-clashing condition for all pairs of balls of radius~$2$ in $\B(G)$ that are centered at vertices in
\begin{equation}\label{eqn:vc:36}
V^W.
\end{equation}
For all $\ell\in [3M]$, $T(u_{\ell},2)\cap W=B_2(u_{\ell})\cap W=W\setminus \{w_{\ell}\}$, and $W\subset T(u_{3M+1},2)$. Hence, $T$ satisfies the non-clashing condition for all pairs of balls of radius~$2$ in $\B(G)$ that are centered at vertices in
\begin{equation}\label{eqn:vc:37}
U \cup \{u_{3M+1}\}.
\end{equation}
For any $x,y\in C \cup C^{\alpha} \cup C^{\beta} \cup C^{\gamma}$, $B_2(x)\setminus (A^{\alpha} \cup A^{\beta} \cup A^{\gamma}) = B_2(y)\setminus (A^{\alpha} \cup A^{\beta} \cup A^{\gamma})$, $B_2(x)\cap (A^{\alpha} \cup A^{\beta} \cup A^{\gamma}) \subset T(x,2)$, and $B_2(y)\cap (A^{\alpha} \cup A^{\beta} \cup A^{\gamma}) \subset T(y,2)$. Hence, $T$ satisfies the non-clashing condition for all pairs of balls of radius~$2$ in $\B(G)$ that are centered at vertices in
\begin{equation}\label{eqn:vc:38}
C \cup C^{\alpha} \cup C^{\beta} \cup C^{\gamma}.
\end{equation}
For any $x,y\in V^{\alpha} \cup V^{\beta} \cup V^{\gamma} \cup V^{\alpha,*} \cup V^{\beta,*} \cup V^{\gamma,*}$, $B_2(x)\setminus (A^{\alpha} \cup A^{\beta} \cup A^{\gamma}) = B_2(y)\setminus (A^{\alpha} \cup A^{\beta} \cup A^{\gamma})$, $B_2(x)\cap (A^{\alpha} \cup A^{\beta} \cup A^{\gamma}) = B_1(x)\cap (A^{\alpha} \cup A^{\beta} \cup A^{\gamma}) \subset T(x,2)$, and $B_2(y)\cap (A^{\alpha} \cup A^{\beta} \cup A^{\gamma}) = B_1(y)\cap (A^{\alpha} \cup A^{\beta} \cup A^{\gamma}) \subset T(y,2)$. Hence, $T$ satisfies the non-clashing condition for all pairs of balls of radius~$2$ in $\B(G)$ that are centered at vertices in
\begin{equation}\label{eqn:vc:39}
V^{\alpha} \cup V^{\beta} \cup V^{\gamma} \cup V^{\alpha,*} \cup V^{\beta,*} \cup V^{\gamma,*}.
\end{equation}
For any $x,y\in A^{\alpha} \cup A^{\beta} \cup A^{\gamma}$, we have that $B_1(x)\subset T(x,2)$, $B_1(y)\subset T(y,2)$, and $B_2(x)\neq B_2(y)$ if and only if $B_1(x)\neq B_1(y)$. Hence, $T$ satisfies the non-clashing condition for all pairs of balls of radius~$2$ in $\B(G)$ that are centered at vertices in
\begin{equation}\label{eqn:vc:40}
A^{\alpha} \cup A^{\beta} \cup A^{\gamma}.
\end{equation}
For all $x\in W$ and $y\in V^W$, $B_2(x)\setminus U = B_2(y)\setminus U$ and $U\subset T(y,2)$. Hence, in combination with (\ref{eqn:vc:35}) and (\ref{eqn:vc:36}), $T$ satisfies the non-clashing condition for all pairs of balls of radius~$2$ in $\B(G)$ that are centered at vertices in
\begin{equation}\label{eqn:vc:41}
W \cup V^W.
\end{equation}
For all $x\in W \cup V^W$ and $y\in V(G)\setminus (W \cup V^W \cup \{u'_{3M+1}\} \cup \{z\})$, $|T(y,2)\cap (A^{\alpha} \cup A^{\beta} \cup A^{\gamma})|\geq 1$ and $|B_2(x)\cap (A^{\alpha} \cup A^{\beta} \cup A^{\gamma})|=0$. Hence, in combination with (\ref{eqn:vc:37})--(\ref{eqn:vc:41}), $T$ satisfies the non-clashing condition for all pairs of balls of radius~$2$ in $\B(G)$ that are centered at vertices in
\begin{equation}\label{eqn:vc:42}
W \cup V^W \cup U \cup \{u_{3M+1}\};
\end{equation}
\begin{equation}\label{eqn:vc:43}
W \cup V^W \cup C \cup C^{\alpha} \cup C^{\beta} \cup C^{\gamma};
\end{equation}
\begin{equation}\label{eqn:vc:44}
W \cup V^W \cup V^{\alpha} \cup V^{\beta} \cup V^{\gamma} \cup V^{\alpha,*} \cup V^{\beta,*} \cup V^{\gamma,*};
\end{equation}
\begin{equation}\label{eqn:vc:45}
W \cup V^W \cup A^{\alpha} \cup A^{\beta} \cup A^{\gamma}.
\end{equation}
For all $x\in U \cup \{u_{3M+1}\}$ and $y\in C \cup C^{\alpha} \cup C^{\beta} \cup C^{\gamma}$, we have that $\pi' \subset T(x,2)$ and $B_2(y)\cap \pi' \neq \pi'$. Indeed, for all $y\in C \cup C^{\alpha} \cup C^{\beta} \cup C^{\gamma}$, the only vertices in $A^{\alpha} \cup A^{\beta} \cup A^{\gamma}$ that are not in $B_2(y)$ are those corresponding to the literals contained in the clause corresponding to $y$ in $\phi$ (as mentioned before, for each $i\in [N]$ and $\delta \in \{\alpha, \beta, \gamma\}$, the vertex $c_i^{\delta}$ can be thought of as a clause containing only the positive and negative literals of $x_i^{\delta}$). The property then follows since $\pi'$ corresponds to the satisfying assignment $\pi$ for $\phi$. Hence, in combination with (\ref{eqn:vc:37}) and (\ref{eqn:vc:38}), $T$ satisfies the non-clashing condition for all pairs of balls of radius~$2$ in $\B(G)$ that are centered at vertices in
\begin{equation}\label{eqn:vc:46}
U \cup \{u_{3M+1}\} \cup C \cup C^{\alpha} \cup C^{\beta} \cup C^{\gamma}.
\end{equation}
For all $x\in U \cup \{u_{3M+1}\} \cup C \cup C^{\alpha} \cup C^{\beta} \cup C^{\gamma}$ and $y\in A^{\alpha} \cup A^{\beta} \cup A^{\gamma}$, $|T(x,2)\cap W|\geq 1$ and $|B_2(y)\cap W|=0$. Hence, in combination with (\ref{eqn:vc:40}) and (\ref{eqn:vc:46}), $T$ satisfies the non-clashing condition for all pairs of balls of radius~$2$ in $\B(G)$ that are centered at vertices in
\begin{equation}\label{eqn:vc:47}
U \cup \{u_{3M+1}\} \cup C \cup C^{\alpha} \cup C^{\beta} \cup C^{\gamma} \cup A^{\alpha} \cup A^{\beta} \cup A^{\gamma}.
\end{equation}
For all $x\in U \cup \{u_{3M+1}\} \cup C \cup C^{\alpha} \cup C^{\beta} \cup C^{\gamma} \cup A^{\alpha} \cup A^{\beta} \cup A^{\gamma}$, $y\in V^{\alpha} \cup V^{\beta} \cup V^{\gamma} \cup V^{\alpha,*} \cup V^{\beta,*} \cup V^{\gamma,*}$, and $\delta \in \{\alpha, \beta, \gamma\}$, $|T(x,2)\cap A^{\delta}|\geq 1$ and, for some $\delta' \in \{\alpha, \beta, \gamma\}$, $|B_2(y)\cap A^{\delta'}|=0$. Hence, in combination with (\ref{eqn:vc:39}) and (\ref{eqn:vc:47}), $T$ satisfies the non-clashing condition for all pairs of balls of radius~$2$ in $\B(G)$ that are centered at vertices in
\begin{equation}\label{eqn:vc:48}
U \cup \{u_{3M+1}\} \cup C \cup C^{\alpha} \cup C^{\beta} \cup C^{\gamma} \cup A^{\alpha} \cup A^{\beta} \cup A^{\gamma} \cup V^{\alpha} \cup V^{\beta} \cup V^{\gamma} \cup V^{\alpha,*} \cup V^{\beta,*} \cup V^{\gamma,*}.
\end{equation}
Combining (\ref{eqn:vc:42})--(\ref{eqn:vc:45}) and (\ref{eqn:vc:48}), $T$ satisfies the non-clashing condition for all pairs of balls of radius~$2$ in $\B(G)$ that are centered at vertices in
\begin{equation}\label{eqn:vc:49}
V(G).
\end{equation}
Combining (\ref{eqn:vc:11}), (\ref{eqn:vc:34}), and (\ref{eqn:vc:49}), we get that $T$ is an $\NCTMp$ of size at most $k$ for $\B(G)$.
\end{proof}

\mainbackward*

\begin{proof}
Suppose that $T$ is an $\NCTMp$ for $\B(G)$ of size $k$. We first prove some properties of $T$. For each $\ell \in [3M]$, to ensure that $T$ satisfies the non-clashing condition for the pair $B_2(u_{3M+1})=V(G)$ and $B_2(u_{\ell})=V(G)\setminus \{w_{\ell}\}$, we have that $|T(u_{3M+1},2)\cap W|=3M$. For each $i \in [N]$ and $\delta \in \{\alpha, \beta, \gamma\}$, to ensure that $T$ satisfies the non-clashing condition for the pair $B_2(u_{3M+1})=V(G)$ and $B_2(c^{\delta}_i)=V(G)\setminus \{t^{\delta}_{2i},f^{\delta}_{2i-1}\}$, we have that $|T(u_{3M+1},2)\cap \{t^{\delta}_{2i},f^{\delta}_{2i-1}\}|\geq 1$. Since $k=3N+3M$, by the two previous arguments, it must be that $|T(u_{3M+1},2)\cap \{t^{\delta}_{2i},f^{\delta}_{2i-1}\}|=1$ for each $i\in [N]$ and $\delta \in \{\alpha, \beta, \gamma\}$.

From $T(u_{3M+1},2)$, we extract an assignment $\pi: X^{\alpha} \cup X^{\beta} \cup X^{\gamma} \rightarrow \{\text{True},\text{False}\}$ for~$\phi$.
For each $i \in [N]$ and $\delta \in \{\alpha, \beta, \gamma\}$, if $T(u_{3M+1},2)\cap \{t^{\delta}_{2i},f^{\delta}_{2i-1}\} = \{t^{\delta}_{2i}\}$, then set $\pi(x^{\delta}_i)=\text{True}$, and otherwise, set $\pi(x^{\delta}_i)=\text{False}$.
Thus, each variable in $\phi$ is assigned exactly one truth value by $\pi$.
It remains to show that $\pi$ is a satisfying assignment for $\phi$.

Recall that, for each $\ell \in [M]$, we have that $c_{\ell}$ is the vertex in $G$ corresponding to the clause $C_{\ell}$ in $\phi$. By the construction, for each $i\in [N]$ and $\delta \in \{\alpha, \beta, \gamma\}$, if $x^{\delta}_i$ appears as a positive (negative, respectively) literal in $C_{\ell}$, then $t^{\delta}_{2i} \notin B_2(c_{\ell})$ ($f^{\delta}_{2i-1} \notin B_2(c_{\ell})$, respectively). Moreover, these are the only vertices of $G$ that are not in $B_2(c_{\ell})$.
Since, for all $\ell \in [M]$, $T$ satisfies the non-clashing condition for the pair $B_2(u_{3M+1})=V(G)$ and $B_2(c_{\ell})$, we have that $T(u_{3M+1},2)$ contains at least one of the vertices missing from $B_2(c_{\ell})$. Further, for each of these vertices in $T(u_{3M+1},2)$ that are missing from $B_2(c_{\ell})$, $\pi$ assigns the corresponding truth value of that vertex to the corresponding variable.
Since this is true for the clause vertices in $C$ corresponding to all the clauses in $\phi$, we have that $\pi$ is a satisfying assignment for $\phi$.
\end{proof}

\kernel*

\begin{proof}
We begin by proving that \srmfull\ admits a kernelization algorithm outputting a kernel with $2^{\mathcal{O}(\vc)}$ vertices. Given a graph $G$, let $X\subseteq V(G)$ be a minimum vertex cover of $G$, that is, $I=:V(G)\setminus X$ is an independent set. If a minimum vertex cover is not given, then we can compute a $2$-approximate vertex cover in polynomial time. The kernelization algorithm exhaustively applies Reduction Rule~\ref{red-rule} to $G$, which is safe by Lemma~\ref{red-rule-safe}. Now, for any instance on which Reduction Rule~\ref{red-rule} cannot be applied, it holds that, for any $Y\subseteq X$, there are at most $2^{\vc}+1$ vertices in $I$ whose open neighborhoods are exactly $Y$. Since there are $2^{\vc}$ distinct subsets of vertices of $X$, there are at most $2^{\vc}\cdot (2^{\vc}+1) + \vc = 2^{\mathcal{O}(\vc)}$ vertices in the reduced instance.

To obtain an algorithm running in time $2^{2^{\mathcal{O}(\vc)}}\cdot n^{\mathcal{O}(1)}$, one can apply the above (polynomial-time) kernelization algorithm to the graph $G$, and then apply the algorithm from Proposition~\ref{exact-algo} to the resulting
kernel.
\end{proof}

\rulesafe*

\begin{proof}
Let $S\subseteq I$ be a set of $2^{|X|}+2$ vertices that are pairwise false twins. Let $T$ be an $\NCTMp$ of size at most $k$ for $\B(G)$.
Since $T$ satisfies the non-clashing condition, for any $u,v\in S$, at least one of the inclusions  $u\in T(u,1), v\in T(v,1)$ holds. Thus, there is
at most one vertex $w\in S$ such that $w\notin T(w,1)$. Note that, for any $u \in S$, $T(u,1) \subseteq B_1(u) \subseteq X\cup \{u\}$.
As there are at most $2^{|X|}$ distinct subsets of the vertices of $X$, and there is at most one vertex $w\in S$ such that $w\notin T(w,1)$,
since $|S|=2^{|X|}+2$, there exist two vertices $x,y\in S$ such that $x\in T(x,1)$, $y\in T(y,1)$, and $T(x,1)\setminus \{x\} = T(y,1)\setminus \{y\}$.
Pick any vertex  $z\in V(G)\setminus \{y\}$ and any $r\in \mathbb{N}$ such that $y\in T(z,r)$.
We assert that removing $y$ from $T(z,r)$ and adding another carefully chosen vertex $v$ to $T(z,r)$ maintains that $T$ is an $\NCTMp$ of size at most $k$ for $\B(G)$.
If it was not the case that $x$ was in $T(z,r)$, then $v=x$, and otherwise, $v$ is any other vertex in $S\setminus \{y\}$ (if $S\setminus \{y\}\subseteq T(z,r)$,
then $y$ is simply removed from $T(z,r)$ and no vertex is added to it).

Namely, let $T'$ be the map obtained from $T$ by applying the above procedure for
all $z\in V(G)\setminus \{y\}$ and any $r\in \mathbb{N}$ such that $y\in T(z,r)$. Note that $x\in T'(z,r)$ and $T(z,r)\setminus \{y\}\subseteq T'(z,r)$.
Clearly, $T'$ has size at most $k$, so it remains to show that $T'$ is an $\NCTMp$ for $\B(G)$.
The presence of $y$ in $T(z,r)$ could only be used to satisfy the non-clashing condition between $B_r(z)$ (which contains $S$) and a ball $B'$ that contains at most $1$ vertex from $S\setminus \{y\}$ since any ball in $G$ contains $0$, $1$ or $|S|$ vertices from $S$ as the vertices of $S$ are pairwise false twins. If $|B'\cap S|=0$, then $T'(z,r)$ satisfies the non-clashing condition for the pair $B_r(z)$ and $B'$ since $x\in T'(z,r)\cap S$.
Otherwise, $|B'\cap S|=1$, and so, $B'$ is a ball of radius~$0$~or $1$ centered at a vertex in $S$. In this case, $T'(z,r)$ clearly satisfies the non-clashing condition for the pair $B_r(z)$ and $B'$, as long as $B'$ is not the ball of radius~$0$~or $1$ centered at $x$. Thus, let $B'\in \{B_0(x),B_1(x)\}$. For $T$ to be an $\NCTMp$ for $\B(G)$, it must be that $T(z,r)\neq \{x\}$, and thus, $T'(z,r)\neq \{x\}$, since otherwise $T$ would not satisfy the non-clashing condition for the pair $B_r(z)$ and $B_0(x)$. Hence, let $B'=B_1(x)$. If there exists $w_1\in T(z,r)$ such that $w_1\notin B'\cup \{y\}$, then $T'$ satisfies the non-clashing condition for the pair $B_r(z)$ and $B'$ since $T(z,r)\setminus \{y\}\subseteq T'(z,r)$.
So, assume no such vertex $w_1$ exists.
Similarly, if there exists $w_2\in T(x,1)$ such that $w_2\notin B_r(z)$, then $T'$ satisfies the non-clashing condition for the pair $B_r(z)$ and $B'$.
So, assume no such vertex $w_2$ exists.
Then, it must be that $T(z,r)\setminus \{y\} \subseteq B_1(x)$ and $T(x,1) \subseteq B_r(z)$.
Thus, $T(y,1) \subseteq B_r(z)$ since $T(x,1)\setminus \{x\}=T(y,1)\setminus \{y\}$ and $x,y\in B_r(z)$.
Since $T$ is an $\NCTMp$ for $\B(G)$, $T(z,r) \setminus B_1(y)\neq \varnothing$.
Let the vertex $s$ be in $T(z,r) \setminus B_1(y)$, and note that $s\in T'(z,r) \setminus B_1(y)$ since $T(z,r)\setminus \{y\}\subseteq T'(z,r)$.
If $s\neq x$, then $T'(z,r)$ satisfies the non-clashing condition for the pair $B_r(z)$ and $B_1(x)$.
If $s=x$, then there exists $t\in T'(z,r)\cap (S\setminus \{x,y\})$ since either $v=t$ or $t$ was already in $T(z,r)$.
In this case, $T'(z,r)$ satisfies the non-clashing condition for the pair $B_r(z)$ and $B_1(x)$.
Consequently, $T'$ is an $\NCTMp$ for $\B(G)$.
Since $y$ is not contained in $T'(z,r)$, then $T'$ restricted to the vertices of $G\setminus \{y\}$ is an $\NCTMp$ of size at most $k$ for $\B(G\setminus\{y\})$.

For the reverse direction, let $T'$ be an $\NCTMp$ of size at
most $k$ for $\B(G\setminus \{y\})$. Without loss of generality, assume that
$T'(z,r) \neq \varnothing$ for any $z \in V(G)$ and any integer $r \geq 0$. First,
note that the addition of  $y$ does not make any two balls
that were the same in $G\setminus \{y\}$ become distinct in $G$.  Indeed, if both balls contained every vertex in
$S\setminus \{y\}$, then they will both contain $y$; if both balls did
not contain any vertex in $S\setminus \{y\}$, then neither of them
will contain $y$; if both balls contained exactly one vertex in
$S\setminus \{y\}$, then neither of them will contain $y$.
Hence, it suffices to extend $T'$ to an $\NCTMp$ $T$ of size at
most $k$ for $\B(G)$, by defining $T(y,r)$ for all $r\in \mathbb{N}$
so that $T$ satisfies the non-clashing condition for any
pair of balls, where one ball  is centered in $y$.  Thus, let
$T(z,r):=T'(z,r)$ for all $z\in V(G)\setminus \{y\}$ and
$r\in \mathbb{N}$. As before, let $x\in S$ be such that $x\in T(x,1)$
(there must be such an $x$ since there is at most one vertex $w\in S$
such that $w\notin T(w,1)$). Set $T(y,0):=\{y\}$,
$T(y,1) := \{y\} \cup (T'(x,1)\setminus \{x\})$, and
$T(y,r) = T'(x,r)$ for all integers $r \geq 2$.  Note that
$B_1(x)\setminus \{x\}=B_1(y)\setminus \{y\}$ and $B_r(x)=B_r(y)$ for
all integers $r\geq 2$.

To show that $T$ is an $\NCTMp$ for $\B(G)$, we have to show that it
satisfies the non-clashing condition for a ball $B\in \{B_0(y),B_1(y)\}$ and any other ball $B' = B_r(z)$.
First, let $z = y$. Since $T'$ satisfies the non-clashing
condition for $B_{0}(x)$ and $B_{r}(x)$ for any $r > 0$, $T'(x,r) \setminus \{x\} \neq \varnothing$, and thus,
$T(y,r) \setminus \{y\} \neq \varnothing$ since
$T'(x,r) \setminus \{x\} \subseteq T(y,r) \setminus \{y\}$ for all
$r > 0$. Moreover, since $T'$ satisfies the non-clashing condition
for $B_1(x)$ and $B_r(x)$ for any $r > 1$,
there exists $u \in T'(x,r) \setminus B_1(x)$. Note that $u \neq y$,
and thus, $u \in T(y,r) \setminus B_1(y)$. Consequently, $T$ satisfies the
non-clashing condition for a ball $B$ in $\{B_0(y),B_1(y)\}$ and any
other ball $B_r(y)$.
Now, let $z = x$.  If $r \leq 1$, then
$y \in T(y,r) \setminus B_r(x)$. If $r \geq 2$, then $B_r(x) = B_r(y)$
and, by the previous case, $T$ satisfies the non-clashing condition for
$B \in \{B_0(y),B_1(y)\}$ and $B_r(x)$.
Finally, let  $z \notin \{x,y\}$.  Since $T(z,r) = T'(z,r)$ is
non-empty and does not contain $y$, we have that
$T(z,r) \setminus B_0(y) \neq \varnothing$. If
$T(z,r) \setminus B_1(y) \neq \varnothing$, then $T$ satisfies the
non-clashing condition for $B_1(y)$ and $B_r(z)$.  Suppose now that
$T(z,r) \subseteq B_1(y)$. Since $y \notin T(z,r) = T'(z,r)$,
$T'(z,r) \subseteq B_1(y) \setminus \{y\} \subseteq B_1(x)$. Since
$T'$ satisfies the non-clashing condition for $B_r(z)$ and $B_1(x)$,
there exists $u \in T'(x,1)\setminus B_r(z)$. If
$u \neq x$, then $u \in T(y,1) \setminus B_r(z)$, and if $u = x$, then
$B_r(z)$ contains neither $x$ nor $y$, and thus,
$y \in T(y,1)\setminus B_r(z)$. In any case, $T$ satisfies the
non-clashing condition for $B_1(y)$ and $B_r(z)$.
Hence, $T$ is an $\NCTMp$ of size at most $k$
for $\B(G)$.
\end{proof}


We continue with Lemma~\ref{min-ball}, which we reformulate and prove in a more general form through the following lemma and its corollary: 

\begin{lemma}{(Lemma~11, \cite{ChChMc})}\label{lem-xinC} Let $B\in \B$, $X$ be a realizable sample for $B$, and $\{ u,v\}$ be a diametral pair of $X^+$ in $\KKK$. If $B_r(x)\in [B]$, $x'$ is the apex of $x$ with respect to $u$ and $v$, and $r'=r-d(x,x')$,
then $X$  is a realizable sample for $B_{r'}(x')$. Consequently, the path of cycles $C(u,v)$ contains a center of a ball realizing $X$. 
\end{lemma}


\begin{corollary} \label{diam2} If $\{ u,v\}$ is a diametral pair of a ball $B_r(x)$, $x'$ is the apex of $x$ with respect to $u$ and $v$, and $r'=r-d(x,x')$, then $B_{r'}(x')=B_r(x)$. In particular, if $B\in \B$, $\{ u,v\}$ is a diametral pair of $B$, and $B_r(x)$ is a minimal ball of $[B]$, then $x$ belongs to $C(u,v)$.
\end{corollary}

\begin{proof} Set $X:=V(\KKK)$. By Lemma~\ref{lem-xinC}, $B_r(x)=X^+\subset B_{r'}(x')\subseteq B_r(x)$, yielding $B_{r'}(x')=B_r(x)$.
\end{proof}

\treesofcycles*

\begin{proof} First, analogous to the case of trees, $T$ is non-clashing for any pair of balls including a ball of radius $0$. Now, suppose that $B_1$ is defined by the minimal ball $B_{r_1}(x)\in [B_1]$ and let $u,v$ be the diametral pair of $B_1$ defining $T^+(B_1)$.
  Analogously, suppose that $B_2$ is defined by the minimal ball $B_{r_2}(y)\in [B_2]$. By contradiction, assume that $B_1\neq B_2$ and that $T$ does not satisfy the non-clashing condition for $B_1$ and $B_2$. Without loss of generality,
suppose that there exists a vertex $z \in B_2\setminus B_1$.

\begin{claim} \label{notartic}
The vertex $x$ is not a cut vertex of $C(u,v)$.
\end{claim}

\begin{proof} By Corollary~\ref{diam2}, $x$ belongs to $C(u,v)$. If $x$ disconnects $u$ and $v$, we can suppose, without loss of generality,
that $x$ also disconnects $z$ and $u$. Consequently, $d(z,u)=d(z,x)+d(x,u)>r_1+d(x,u) \geq d(v,x)+d(x,u)=d(v,u)$.
Thus, $\diam(B_2)\geq d(u,z)>d(u,v)=\diam(B_1)$, contrary to the assumption that $B_1$ and $B_2$ have the same diameter.
\end{proof}

Consequently, $x$ belongs to a unique (gated) cycle $C$ of $C(u,v)$. Let $y'$ and $z'$ be
the gates of $y$ and $z$ in $C$.  Recall also that $u'$ and $v'$ are the gates of $u$ and $v$ in $C$.
Since $\diam(B_1)=\diam(B_2)$ and $T^+(B_1)\subset B_2$, $u,v$ is also a diametral pair of $B_2$.  By Corollary~\ref{diam2}
applied to $B_2=B_{r_2}(y)$ and the diametral pair $u,v$ of $B_2$, we conclude  that $y\in C(u,v)$.
Therefore, either $y$ belongs to the cycle $C$ or the gate $y'$ of $y$ in  $C$ coincides with $u'$ or $v'$.

\begin{claim} \label{notinterval}
$z' \notin I(x,u') \cup I(x,v')$.
\end{claim}

\begin{proof} Suppose by way of contradiction that $z' \in I(x,u')$. Then, $z'$ and $v'$ separate $v$ and $z$,
and thus, $d(v,z)=d(v,v')+d(v',z')+d(z',z)$. First, suppose that $x \in I(v',z')$. Then, $x\in I(v,z)$.
Since $u\in B_{r_1}(x)$ and $z\notin B_{r_1}(x)$, we obtain that $d(v,u)\leq d(v,x)+d(x,u)<d(v,x)+d(x,z)=d(v,z)$.
Consequently, $\diam (B_2)>\diam (B_1)$, a contradiction. Now, suppose that
$x \notin I(v',z')$. This implies that $u'\in I(z',v')$. Since $u,v$ is a diametral pair of $B_2$, we obtain
that $d(v,v')+d(v',u')+d(u',z')+d(z',z)=d(v,z)\leq d(v,u)=d(v,v')+d(v',u')+d(u',u)$, yielding
$d(u',z')+d(z',z) \leq d(u',u)$. Since $z'\in I(x,u')$, $u\in B_{r_1}(x)$, and $z\notin B_{r_1}(x)$, we obtain that
$d(x,z')+d(z',u')+d(u',u)=d(x,u) \leq r_1<d(x,z)=d(x,z')+d(z',z)$, yielding  $d(z',u')+d(u',u)<d(z',z)$.
From the inequalities $d(u',z')+d(z',z) \leq d(u',u)$ and $d(z',u')+d(u',u)<d(z',z)$, we obtain that $d(u',z')<0$, a contradiction.
\end{proof}

Since $x$ is the apex of $x$ with respect to $u$ and $v$, the shortest
paths $I(x,u')$ and $I(x,v')$ intersect only in $x$, and their union
is the $(u',v')$-path passing via $x$.  By Claim~\ref{notinterval},
$z'$ belongs to the complementary $(u',v')$-path $P$ of $C$ and
$I(x,z')\cap \{ u',v'\}\ne \varnothing$. Hence, the set
$Z(x,u,v)$ is non-empty. Indeed, let $w$ be a vertex of $I(x,z)$ at
distance $r_1+1$ from $x$. Then, either $z'$ is the gate of $w$ in $C$
or $w$ is a vertex of the path $P$, showing that $w\in
Z(x,u,v)$. Consequently, $Z(x,u,v)\ne\varnothing$, and thus, at least
one of the vertices $s,t$ exists. If $s$ ($t$, respectively) exists,
then its gate $s'$ ($t'$, respectively) in $C$ belongs to the path $P$. 
Since $T$ does not satisfy the non-clashing condition for $B_1$ and $B_2$, if $s$ ($t$, respectively) exists, 
then $s \notin B_2$ ($t \notin B_2$, respectively), and thus, $z \neq s$ ($z \neq t$, respectively).

\begin{claim}
  $z' \in I(u',s')\cup I(v',t')$
\end{claim}

\begin{proof}
  Suppose first that $C \nsubseteq B_1$. If the claim does not hold,
  then $z'$ belongs to the $(s',t')$-path of $C$ that does not contain
  $x$.  This implies that removing $s=s'$ and $t=t'$ disconnects $z$
  and $u$. Since $u,z \in B_2$ and $s,t \notin B_2$, we obtain a
  contradiction.

  Suppose now that $C \subseteq B_1$.  Since
  $z' \notin I(x,u') \cup I(x,v')$, $u'$ and $v'$ disconnect $x$ from
  $z'$ and~$z$. Consequently, $I(x,z) \cap \{u',v'\} \neq \varnothing$.
  Without loss of generality, assume that $u' \in I(x,z')$ and
  consider a vertex $w \in I(x,z)$ such that $d(x,w) = r_1+1$. Since
  $C \subseteq B_1$, $w \notin C$ and $z'$ disconnects $x$ and $w$,
  \textit{i.e.}, $w' = z'$. Therefore, $w \in Z^u(x,u,v)$, and thus, by the
  definition of $s$, we have $d(u',s') \geq d(u',w')$. Consequently,
  $w'=z' \in I(u',s')$, and we are done.
\end{proof}

Without loss of generality, assume that $z' \in I(u',s')$.  By Claim 
\ref{notinterval}, $z'\ne u'$. Since
$u'\in I(z',x)\subseteq I(s',x)$, $z'\in I(x,z)$, and $s'\in I(x,s)$,
we obtain that $d(x,z)=d(x,u')+d(u',z')+d(z',z)$ and
$d(x,s)=d(x,u')+d(u',z')+d(z',s')+d(s',s)$. Since
$d(x,z)\geq r_1+1=d(x,s)$, from the two previous equalities we
conclude that $d(z',z)\geq d(z',s')+d(s',s)$. Recall that the vertex
$y$ belongs to $C(u,v)$, and thus, either $y=y'\in C$ or
$y' \in \{u',v'\}$. Consequently, $d(y,z)=d(y,y')+ d(y',z')+d(z',z)$.
Since $s\notin B_2$,
$r_2<d(y,s)=d(y,y')+d(y',s')+d(s',s) \leq
d(y,y')+d(y',z')+d(z',s')+d(s',s) \leq d(y,y')+d(y',z')+d(z',z)=d(y,z)
\leq r_2$, a contradiction (by the triangle inequality and the
inequality $d(z',s')+d(s',s)\le d(z',z)$ established above). This
contradiction establishes that under the conditions of the
theorem, we must have $B_1=B_2$.

The second assertion is a consequence of the first one since two balls with distinct diameters are
distinguished by the diametral pair of the ball with the larger diameter. Finally, $\VCD(\B(\KKK)) \leq 3$ follows from the fact that
trees of cycles $\KKK$ cannot be contracted to $K_4$ and the result of \cite{BouTh,ChEsVa} that if a graph $G$ does not contain $K_{d+1}$ as a
minor, then $\VCD(\B(G))\le d$.
\end{proof}

We continue with the proof of the theorem for $\delta$-hyperbolic graphs.

\hyperbolic*

\begin{proof}
  The proof is similar to the proof for trees.  Clearly,
  $T(x,r)\subseteq B_r(x)$ for any vertex $x$ and any radius $r$, and
  thus, the map $T$ satisfies the inclusion condition.
  Hence, since $|T(x,r)|=2$ for any ball $B_r(x)$ with $r\geq 1$, $T$ is non-clashing for any pair of balls that includes a ball of radius $0$. Now, assume that
  $B_{r_1}(x)$ and $B_{r_2}(y)$ are not $\delta$-identical, and
  suppose, without loss of generality, that there exists
  $z \in B_{r_2}(y) \setminus B_{r_1+2\delta}(x)$. Let
  $T(y,r_2) = \{u,v\}$ and note that, since $\HHH$ is $\delta$-hyperbolic, we have
  $d(x,z) + d(u,v) \leq \max \{d(x,u) + d(v,z), d(x,v) + d(u,z) \} +
  2\delta$. Without loss of generality, assume that
  $d(x,z) + d(u,v) \leq d(x,v) + d(u,z) + 2\delta$. Since
  $u,z \in B_{r_2}(y)$ and since $\{u,v\}$ is a diametral pair of
  $B_{r_2}(y)$, $d(u,z) \leq d(u,v)$. Consequently,
  $r_1 + 2\delta < d(x,z) \leq d(x,v) + 2 \delta$, and thus,
  $v \notin B_{r_1}(x)$. Therefore,
  $v \in (B_{r_2}(y) \setminus B_{r_1}(y)) \cap (T(x,r_1) \cup T(y,r_2))$,
  establishing that $T$ satisfies the non-clashing condition
  for $B_{r_1}(x)$ and $B_{r_2}(y)$.
\end{proof}
\end{document}